\documentclass[a4paper,11pt,reqno]{amsart}

\usepackage{amssymb}
\usepackage{latexsym}
\usepackage{amsmath}
\usepackage{graphicx}
\usepackage{amsthm}
\usepackage{empheq}
\usepackage{bm}

\usepackage[dvipsnames]{xcolor}
\usepackage{pagecolor}
\usepackage{subcaption}
\usepackage{tikz,pgfplots}
\pgfplotsset{compat=newest}

\usepackage[margin=2.3cm]{geometry}

\numberwithin{equation}{section}

\theoremstyle{definition}
\newtheorem{theorem}{Theorem}[section]
\newtheorem{corollary}[theorem]{Corollary}
\newtheorem{proposition}[theorem]{Proposition}
\newtheorem{definition}[theorem]{Definition}
\newtheorem{example}[theorem]{Example}
\newtheorem{notation}[theorem]{Notation}
\newtheorem{remark}[theorem]{Remark}
\newtheorem{lemma}[theorem]{Lemma}

\newtheorem*{thnonumber}{Theorem}

\newcommand{\inn}{\textnormal{in}}
\newcommand\qbin[3]{\left[\begin{matrix} #1 \\ #2 \end{matrix} \right]_{#3}}
\newcommand{\cal}{\mathcal}

\newcommand{\numberset}{\mathbb}
\newcommand{\N}{\numberset{N}}
\newcommand{\Z}{\numberset{Z}}
\newcommand{\Q}{\numberset{Q}}
\newcommand{\R}{\numberset{R}}

\newcommand{\F}{\numberset{F}}

\newcommand{\fq}{\F_q}

\newcommand{\HH}{\textnormal{H}}

\newcommand{\mH}{\mathcal{H}}
\newcommand{\mL}{\mathcal{L}}

\newcommand{\Pro}{\numberset{P}}
\newcommand{\mC}{\mathcal{C}}
\newcommand{\mS}{\mathcal{S}}

\newcommand{\mD}{\mathcal{D}}
\newcommand{\mF}{\mathcal{F}}
\newcommand{\rk}{\textnormal{rk}}
\newcommand{\mB}{\mathcal{B}}
\newcommand{\mP}{\mathcal{P}}
\newcommand{\mO}{\mathcal{O}}
\newcommand{\mat}{\F_q^{n \times m}}

\newcommand{\Ball}{B}
\newcommand{\ball}{\bm{b}}
\newcommand{\bH}{\ball^\HH}
\newcommand{\brk}{\ball^\rk}
\newcommand{\dH}{d^{\textnormal{H}}}
\newcommand{\wH}{\omega^{\textnormal{H}}}
\newcommand{\drk}{d^{\textnormal{rk}}}
\newcommand{\rhork}{\rho^{\textnormal{rk}}}
\newcommand{\rhoH}{\rho^{\textnormal{H}}}
\newcommand{\wrk}{\omega^{\rk}}
\newcommand{\WH}{W^{\HH}}
\newcommand{\Wrk}{W^{\rk}}

\newcommand{\pp}{\bm{p}}
\newcommand\p[3]{\pp(#1;#2,#3)}
\newcommand\pH[3]{\pp^\HH(#1;#2,#3)}
\newcommand\DD[2]{|#1| / |#2|}

\newtheorem{claim}{Claim}

\newcommand*{\myproofname}{Proof of the claim}
\newenvironment{clproof}[1][\myproofname]{\begin{proof}[#1]}{\end{proof}}

\usepackage{hyperref}

\begin{document}

\title[Partition-Balanced Families of Codes and Asymptotic Enumeration]{
Partition-Balanced Families of Codes \\ and Asymptotic Enumeration in Coding Theory}

\author[Eimear Byrne]{Eimear Byrne}
\address{School of Mathematics and Statistics, University College Dublin, Belfield, Ireland}
\curraddr{}
\email{ebyrne@ucd.ie}
\thanks{}

\author[Alberto Ravagnani]{Alberto Ravagnani$^*$}
\address{School of Mathematics and Statistics, University College Dublin, Belfield, Ireland}
\curraddr{}
\email{alberto.ravagnani@ucd.ie}
\thanks{$^*$The author was partially supported by the Swiss National Science Foundation through grant n. P2NEP2\_168527 and by the Marie Curie Research Grants
Scheme, grant n. 740880.}



\maketitle

\thispagestyle{empty}

\begin{abstract}
We introduce the class of partition-balanced families of codes,
and show how to exploit their combinatorial invariants to obtain upper and lower bounds on the number of codes that have a prescribed property. 
In particular, we derive precise asymptotic estimates on the density functions of several classes of codes that are extremal with respect to minimum distance, covering radius, and maximality.
The techniques developed in this paper apply to various distance functions, including the Hamming and the rank metric distances.
Applications of our results show that, unlike the $\F_{q^m}$-linear MRD codes, the $\fq$-linear MRD codes are not dense in the family of codes of the same dimension. More precisely, we show that the density of $\fq$-linear MRD codes in $\F_q^{n \times m}$ in the set of all matrix codes of the same dimension is asymptotically at most $1/2$, both as $q \to +\infty$ and as $m \to +\infty$. We also prove that MDS and $\F_{q^m}$-linear MRD codes are dense in the family of maximal codes. 
Although there does not exist a direct analogue of the redundancy bound for the covering radius of $\F_q$-linear rank metric codes, we show that a similar bound is satisfied by a uniformly random matrix code with high probability. In particular, we prove that codes meeting this bound are dense. 
Finally, we compute the average weight distribution of linear codes in the rank metric, and other parameters that generalize the total weight of a linear code.  
\end{abstract}

\bigskip

\bigskip

\section*{Introduction}\label{sec:intro}

Linear codes over finite fields have been extensively studied as combinatorial objects, with connections to many areas in mathematics such as graph theory, curves over finite fields,
finite geometry, lattice theory and numerous topics in algebraic combinatorics. See, for example ~\cite{caldkant,cameronvanlint,duursma,macws,vLW} and the references therein.

There are several fundamental parameters and invariants associated with a linear code, such as its dimension, minimum distance, covering radius and weight distribution. Determination of some or all of these parameters is a non-trivial problem for an arbitrary code, especially as the dimension of its ambient space increases.
For this reason, constructions of particular classes and families of codes with prescribed parameter sets are often sought. 
Much research has been spent on developing coding theoretic bounds as functions of some of the code parameters. Codes that meet such bounds are \textit{extremal} and highly interesting from a combinatorial point of view, as they often have remarkable rigidity properties. 
    
    In this paper, we offer a new perspective on extremal codes. We obtain upper and lower bounds on the {\em density functions} of a number of families of codes within a larger family, and give precise asymptotic estimates of these. We introduce the idea of a {\em partition-balanced} family of codes, and show how the combinatorial invariants of such families can be used to obtain estimates on the number of codes satisfying a particular property.

As we will show, these techniques can be applied in different contexts to establish the density or sparsity of families of codes that are extremal with respect to {\em minimum distance}, {\em covering radius}, and the related concept of code {\em maximality}.
 Our methods can be used to study codes whose ambient space is endowed with a large class of distance functions. 
 In this paper, we focus on two major distance functions as applications of our results, namely the Hamming and the rank metric. 
    
    The Hamming metric is the classical distance function associated with coding theory. The linear {\em maximum distance separable} (MDS) codes are those $k$-dimensional subpaces of $\fq^n$ meeting the {\em Singleton bound}, namely those whose minimum Hamming distance is exactly $n-k+1$. 
 There is still substantial activity around such codes, which form a central topic in coding theory.

    Another important distance function of coding theory is the {\em rank metric}, which measures the rank of the difference between a pair of matrices with entries from a finite field $\fq$.
    Rank metric codes have seen a recent resurgence of interest both for their potential use in code based cryptography and as error-correcting codes in network communications \cite{gok,l,kk,skk,wzvs,wzpr}. They are also intriguing as mathematical objects in their own right, and several researchers have sought to describe their structural properties \cite{antr_heide,bcbds,byrneravagnani,duallypaper,dlCKWW,del,gad_paper,psgorlaravagnani,ln,alb1,sheekey}. However, the general theory of rank metric codes is still rather unexplored.
    The rank metric analogue of the Singleton bound yields the class of {\em maximum rank distance} (MRD) codes, which exist for all choices of $m,n$ and minimum $\fq$-rank $d$, both for 
    $\F_{q^m}$-linear subspaces of $\F_{q^m}^n$ (which we will refer to as vector rank metric codes) and the larger class of $\fq$-linear subspaces of $\mat$ (which we will refer to as matrix rank metric codes).
    
    Concrete realisations of $\F_{q^m}$-linear MRD codes have been known since the 1970s, having been independently introduced by
    Delsarte, Gabidulin and Roth 
    who studied them from different perspectives~\cite{del,gabid,roth}. On the other hand, general classes of $\fq$-linear MRD codes that are not $\F_{q^m}$-linear were unknown until
    Sheekey~\cite{sheekey} introduced the family of twisted Gabidulin codes.
    
    While the vector rank metric codes often exhibit a behaviour similar to block codes with the Hamming distance, there is considerable divergence between these families and the class of matrix rank metric codes: if similar techniques for Hamming metric codes can be applied to make statements on vector rank metric codes, such methods often fail for matrix rank metric codes. Several examples of this can be observed in this work. 
    
    Perhaps the most profound difference is to be seen in the behaviour of the density
    functions of codes that are extremal with respect to the minimum distance. As the reader will see, while both MDS and vector rank metric MRD codes are dense among codes having the same dimension, the matrix MRD codes are {\em never} dense in this sense, both as $q \to +\infty$ and as $m \to +\infty$\footnote{In the very final stages of writing this paper, we became aware of the preprint~\cite{antr_heide}, in which the authors independently show, by a different argument, that the MRD matrix codes are not dense in the set of codes with the same dimension as $q \to +\infty$.}.
More precisely, one of the results of this paper is the following (see Theorem \ref{bigthm} and its corollaries).

\begin{thnonumber}
Fix integers $2 \le d \le n$. Given $m \in \N$ with $m \ge n$ and a prime power $q$, denote by 
$N_{q,m}$ the number of rank metric codes in $\F_q^{n \times m}$ of dimension $m(k-d+1)$, and by 
$N'_{q,m}$ the number of such codes of minimum distance at most $d-1$. 
Then for every positive real number $\varepsilon$ 
there exists $q_\varepsilon \in \N$ such that
$$\frac{N'_{q,m}}{N_{q,m}} \ge \frac{1}{2}-\varepsilon \qquad \mbox{for all $m \ge n$ and $q \ge q_\varepsilon$.}$$
Moreover, for every positive real number $\varepsilon$  
there exists $m_\varepsilon \in \N$ such that
$$\frac{N'_{q,m}}{N_{q,m}} \ge \frac{1}{2} \left( \frac{q}{q-1} - (q-1)^{-2} \right)-\varepsilon \qquad \mbox{for all $q$ and $m \ge m_\varepsilon$.}$$
\end{thnonumber}

We obtain several other estimates on the density of codes that are extremal with respect to minimum distance, covering radius and maximality, which we outline below.
Standard methods attempting to address density questions in coding theory often rely on the Schwartz-Zippel Lemma~\cite{schwartz,zippel} to obtain lower bounds on density functions. However, as the reader will see, in some important cases these methods fail. Our techniques offer an alternative general approach to asymptotic enumeration problems in coding theory.

 \subsection*{Outline.} In Section~\ref{sec:distregsp} we define basic concepts of distance-regular spaces and their codes. In Section~\ref{sec:part}, we introduce the concept of a partition-balanced family of codes, with respect to an arbitrary partition of the ambient space. We compute the invariants associated with some of these families, which will be used several times throughout the paper. In Section~\ref{sec:densityfs} we define the density functions associated with a family of codes, and give asymptotic estimates of functions that are used in later sections. 
 
 In Section~\ref{sec:densmindist} we give precise asymptotic estimates for the number of codes with the Hamming and the rank metric having given dimension and minimum distance. As immediate corollaries, we obtain the density of MDS and MRD vector rank metric codes in their respective ambient spaces. In Section~\ref{sec:mrdnotsense} we show that the matrix MRD codes are not dense in the family of matrix codes of the same dimension. In particular, we show that the density function of non-MRD matrix codes of dimension $m(n-d+1)$ in $\F_{q}^{m \times n}$ is asymptotically lower bounded by $1/2$.
    
    In Section~\ref{sec:densmaxl}, we show that the MDS and vector rank metric MRD codes are dense in the family of maximal codes for the same dimension. 
    In Section~\ref{denscovrad} we show that Hamming and vector rank metric codes meeting the {\em redundancy bound} are dense in the family of codes of the same dimension. We introduce a new upper bound on the covering radius of matrix rank metric codes, which is not in fact satisfied by all rank metric codes, but rather by a uniformly random code with high probability. We then show that the matrix codes satisfying this bound are dense in the family of matrix codes of the same dimension. 
    
    Finally, in Section~\ref{sec:avpar} we compute the average weight distributions of Hamming metric, vector rank metric and matrix rank metric codes. We obtain asymptotic estimates of these values and observe the interesting fact that, although the MDS and vector rank metric MRD codes are dense as the field size grows, the number of words of weight $d-1$ in a uniformly random linear code converges to a non-zero constant.

\section{Distance-Regular Spaces and Codes}\label{sec:distregsp}

We start by describing
 the class of metric spaces for which our methods apply. These are linear spaces defined over finite fields that exhibit  certain 
regularity properties with respect to their distance functions. They  include important examples from coding theory, such as the  Hamming metric and rank metric spaces. 

\begin{definition} \label{defregular}
Let $Q$ be a prime power, and let $\F_Q$ be the finite field with $Q$ elements. 
Let $X$ be a finite-dimensional vector space over $\F_Q$ and let $d:X \times X \to \N$ be an integer-valued distance function on $X$.
We say that $(X,d)$ is a \textbf{$Q$-ary distance-regular space} if:
\begin{enumerate} 
\item \label{pr3} for all $x,y \in X$ and all $\alpha \in \F_Q \setminus \{0\}$ we have $d(\alpha x, \alpha y) = d(x,y)$,
\item \label{pr3a} for all $x,y \in X$ we have $d(x, y) = d(0,x-y)$,
\item \label{pr4} for all $x,y \in X$ and $j,k \in \N$, 
$|\{z \in X \mid d(z,x)=j, \ d(z,y)=k\}|$ only depends on 
$d(x,y)$.
\end{enumerate}
Denote by
$\omega:X \to \N$ the \textbf{weight} induced by $d$, i.e., the function defined by
$\omega(x):=d(x,0)$ for all $x \in X$. To simplify the discussion in the sequel, we also assume that a $Q$-ary distance-regular space $(X,d)$ with weight $\omega$ satisfies: 
\begin{enumerate} 
\addtocounter{enumi}{3}
\item $|\omega|:=\max\{\omega(x) \mid x \in X \} =|\omega(X)|-1$, \label{pr5}
\end{enumerate}
where $\omega(X)=\{\omega(x)\mid x \in X\}$ is the image of $X$ under $\omega$. 
\end{definition}

Note that a $Q$-ary distance-regular space is an example of a {\em symmetric association scheme}. 
Therefore, for each triple $i,j,k \in \N$ there is an associated \textbf{intersection number} of $(X,d)$, defined to be the integer
$$\p{i}{j}{k}:= |\{z \in X \mid d(z,x)=j, \ d(z,y)=k\}|,$$
where $x,y \in X$ are any vectors with $d(x,y)=i$. 
The properties of the $\p{i}{j}{k}$'s are well-studied. The interested reader is referred to~\cite{BCN} for further details.

\begin{remark} \label{rem:candivide}
	By Property~\ref{pr3} of Definition~\ref{defregular} we have $\omega(\alpha x)=\omega(x)$ for all $x \in X$ and $\alpha \in \F_Q \setminus \{0\}$. Moreover, Property~\ref{pr5} implies that for all
	$0 \le i \le |\omega|$ there exists $x \in X$ with $\omega(x)=i$. 
\end{remark}

\begin{notation}
For the remainder, $Q$ denotes a prime power, and $(X,d)$ a fixed $Q$-ary distance-regular space of dimension $N \ge 2$ over $\F_Q$. 
\end{notation}

We are interested in the combinatorial properties and invariants of the subsets of $X$. Our focus will be mostly on the subspaces. 
We define several of the coding theoretic invariants that we will consider in this paper.

\begin{definition}
A \textbf{code} is a non-empty subset $\mC \subseteq X$, and its elements are its \textbf{codewords}. 
We say that $\mC$ is \textbf{linear} 
if it is an $\F_Q$-linear subspace of $X$. In this case we write $\mC \le X$.

Let $\mC \subseteq X$ be a (not necessarily linear) code. 
 If $|\mC| \ge 2$, then the \textbf{minimum distance} of $\mC$ is the integer
$d(\mC):=\min\{d(x,y) \mid x,y \in \mC,\ x \neq y\}$. We also set $d(\{0\}):=+\infty$.
The \textbf{weight distribution} of $\mC$ is the sequence $(W_i(\mC) \mid i \in \N)$, where
$W_i(\mC):=|\{x \in \mC \mid \omega(x)=i\}|$ for all $i \in \N$.
Finally, the \textbf{covering radius} of $\mC$ is the integer 
$\rho(\mC):=\min\{r \in \N \mid \mbox{for all } x \in X \mbox{ there exists } y \in \mC \mbox{ with } d(x,y) \le r\}$.
\end{definition}

Note that if $\mC \le X$ is a non-zero linear code, then it follows immediately from the definition of minimum distance and the linearity of $\mC$ that $d(\mC)=\min\{W_i(\mC) \mid i \in \N, \ i \neq 0\}$.

A {\em ball} in $(X,d)$ of given radius, say $r$, is an example of a code that contain only codewords of weight at most $r$. Clearly, a linear code $\mC$ will intersect a ball of radius $r$ centred at zero only if the minimum distance of $\mC$ is at most $r$. We will apply this observation later to obtain estimates on the  density of families of codes characterized as having certain properties.

\begin{definition}
Let $x \in X$ and $r \in \N$. The \textbf{ball} of radius $r$ and center $x$ is the set
$$\Ball(x,r):=\{y \in X \mid d(x,y) \le r\} \subseteq X.$$
\end{definition}

The size of the ball $\Ball(x,r)$ only depends on $r$.
This follows easily from the definition of $\Ball(x,r)$ and Property~\ref{pr4} of Definition~\ref{defregular}. More precisely, for all $x \in X$ and $r \in \N$ we have
$$|\Ball(x,r)|=\sum_{j=0}^{r} \p{0}{j}{j}.$$

\begin{notation}
For $r \in \N$, we denote by $\ball(r)$ the cardinality of $|\Ball(x,r)|$, for any vector $x \in X$.
\end{notation}

\section{Partition-Balanced Families of Codes}\label{sec:part}

We describe families of codes that exhibit regularity properties with respect to a given partition of the ambient space $X$.

\begin{notation}
Let $\mF$ be a \textbf{family} of codes in $X$, i.e., a collection of non-empty subsets of $X$. For $x \in X$, we let
$$\mF_x:=\{ \mC \in \mF : x \in \mC \}.$$
\end{notation}

We are interested in families $\mF$ such that the cardinality of $\mF_x$ depends only on the \textit{class} of $x$ with respect to a given partition, say $\mP$,  of the ambient space $X$. 
This motivates the following definition.

\begin{definition}\label{def:pbal}
Let $\mP=\{\mP_1,...,\mP_M\}$ be a partition of $X$ of size $M$. A non-empty family $\mF$ of codes in $X$ is called 
\textbf{$\mP$-balanced} if
$|\mF_x|$ depends only on the integer $i$ such that $x \in \mP_i$, for all $x \in X$.
In words, the family $\mF$ is $\mP$-balanced if the number of codes $\mC \in \mF$ containing $x \in X$
only depends on the class of $\mP$ containing $x$. 
If $\mF$ is $\mP$-balanced, then 
the \textbf{invariants} of the pair $(\mP,\mF)$ are the integers defined by
$$\mP_i(\mF):=|\mF_x|=|\{\mC \in \mF \mid x \in \mC\}| \quad \mbox{for } 1 \le i \le M,$$
where $x \in X$ is any element with $x \in \mP_i$.
\end{definition}

In~\cite{pelbook}, the authors give a definition of a balanced family of codes that is a special case of Definition \ref{def:pbal}. In particular, it is defined with respect to the partition of $X$ into two classes, namely, $\mP_1 = \{0\}$ and $\mP_2 = X \backslash \{0\}$. As the reader will see, for our results we need to consider more general partitions on $X$.


A useful property of partition-balanced families is the following simple identity, which will play a crucial role throughout the paper.

\begin{lemma} \label{KEY}
Let $\mP$ be a partition of $X$ of size $M$, and let $\mF$ be a $\mP$-balanced family of codes in $X$.
Then for all functions $f:X \to \R$ we have
$$\sum_{\mC \in \mF} \sum_{x \in \mC} f(x) = \sum_{i=1}^M \mP_i(\mF) \sum_{x \in \mP_i} f(x).$$
\end{lemma}

\begin{proof}
Exchanging the order of summation we obtain
$$\sum_{\mC \in \mF} \sum_{x \in \mC} f(x) =  \sum_{x \in X} \sum_{\substack{\mC \in \mF \\ x \in \mC}} f(x)=
\sum_{x \in X} |\mF_x| \cdot f(x) =  \sum_{i=1}^M \mP_i(\mF) \sum_{x \in \mP_i} f(x),$$
as desired.
\end{proof}

\begin{remark}
	In the case that the function $f: X \to \R$ of Lemma~\ref{KEY} is the characteristic function of a set $\mS \subseteq X$,  Lemma~\ref{KEY} yields
	$$\sum_{\mC \in \mF} |\mC \cap \mS|= \sum_{i=1}^M \mP_i(\mF) \cdot |\mP_i \cap \mS|.$$
	In particular, it expresses the average intersection between $\mS$ and a code in $\mF$ in terms of $|\mF|$, the invariants of $(\mP,\mF)$, and the intersections between $\mS$ with the classes of $\mP$.
	We will apply often arguments of this type.  
\end{remark}

We now consider specific partitions and families of codes that are balanced with respect to these partitions. We explicitly compute the invariants of such partition-family pairs. These invariants will be required in later in order to make statements on the density of certain classes of codes.
More precisely, Proposition~\ref{examp1} will be used in Section~\ref{sec:densmindist} to study the density of the MDS codes and MRD vector rank metric codes, and again in Sections~\ref{sec:densmaxl} and~\ref{sec:avpar}.
Proposition~\ref{examp3} will be applied in Section~\ref{sec:mrdnotsense} to establish the remarkable fact that $\F_q$-linear MRD codes are not dense in $\F_{q}^{m \times n}$, both as
$q \to +\infty$ and as $m \to +\infty$.

\begin{proposition} \label{examp1}
Let $\mD \le X$ be a linear code of dimension $t<N$. Construct a partition $\mP$ of $X$ of size two via
$\mP_1:=\mD$ and $\mP_2:=X \setminus \mD$. Fix an integer $k$ with $t \le k \le N$, and define the family
$\mF:=\{\mC \le X \mid \mD \le \mC, \ \dim(\mC)=k\}$. Then
 $\mF$ is $\mP$-balanced. Moreover,
$$|\mF| = \qbin{N-t}{k-t}{Q}, \qquad \mP_1(\mF) = |\mF|, \qquad 
\mP_2(\mF)=\frac{|\mF| \cdot (Q^k-Q^t)}{Q^N-Q^t}.$$
\end{proposition}

\begin{proof}
We first show that $\mF$ is $\mP$-balanced. Assume that $x,y \in X$ are in the same class of $\mP$.
If $x,y \in \mP_1$, then all codes $\mC \in \mF$ contain both $x$ and $y$. If $x,y \in \mP_2$, then it is easy to see that there exists an $\F_Q$-linear isomorphism $G:X \to X$ such that $G(\mD)=\mD$ and $G(x)=y$. Such a map $G$ induces a bijection between codes 
in $\mF$ containing $x$ and codes in $\mF$ containing $y$.

The formulas for $|\mF|$ and $\mP_1(\mF)$ are immediate.
To compute $\mP_2(\mF)$, it suffices to count the elements of the set
$\{(x,\mC) \in \mP_2 \times \mF \mid \mC \in \mF, \ x \in \mC\}$ in two different ways, which gives the identity
$|\mP_2| \cdot \mP_2(\mF) = |\mF| \cdot (Q^k-Q^t)$.
\end{proof}

A natural partition of $X$ is the one induced by the weight function $\omega:X \to \N$.

\begin{definition}  \label{def:WB}
For all $i \in \N$, let $\mP_i(\omega):=\{x \in X \mid \omega(x)=i\}$.
 The \textbf{weight partition} of $X$, denoted by $\mP(\omega)$,  is the partition whose 
classes are the non-empty sets of the form $\mP_i(\omega)$, $i \in \N$. Note that
$\mP(\omega)$ is a partition of size $|\omega|+1$ by Property~\ref{pr5} of Definition~\ref{defregular}. 
\end{definition}

The invariants of a weight partition balanced family can be computed as follows.

\begin{proposition}  \label{examp2}
Assume that $\mF$ is a $\mP(\omega)$-balanced family of linear codes in $X$. We have
$$\mP_i(\omega)(\mF)=\frac{\sum_{\mC \in \mF} W_i(\mC)}{W_i(X)} \ \mbox{ for all } 0 \le i \le |\omega|.$$
\end{proposition}

\begin{proof}
It suffices to double-count the elements of 
$\{(\mC,x) \in \mF \times X \mid \mC \in \mF, \ x \in \mC, \ \omega(x)=i\}$ and then use Remark~\ref{rem:candivide}, which guarantees that
$W_i(X) \neq 0$ for all $0 \le i \le |\omega|$. 
\end{proof}

We conclude this section with another class of partition balanced families.

\begin{proposition} \label{examp3}
Let $\mD \le X$ be an $\F_Q$-linear code of dimension $1 \le t < N$. Let $\mP$ be the partition of $X$ of size three given by 
$\mP_1:=\{0\}$, $\mP_2:=\mD \setminus \{0\}$, and $\mP_3:=X \setminus \mD$. Fix an integer $1 \le k \le N$, and define the family
$\mF_\mD:=\{\mC \le X \mid \dim(\mC)=k, \ \mC \cap \mD \neq \{0\}\}$. Then $\mF_\mD$ is $\mP$-balanced. 
Moreover,
$$\mP_1(\mF_\mD)=|\mF_\mD|, \qquad \mP_2(\mF_\mD)=\frac{\sum_{\mC \in \mF_\mD} |\mC \cap (\mD \setminus \{0\})|}{Q^t-1}, \qquad\mP_3(\mF_\mD)=\frac{\sum_{\mC \in \mF_\mD} |\mC \cap (X \setminus \mD)|}{Q^N-Q^t}.$$
\end{proposition}

\begin{proof}
By definition, to see that $\mF_\mD$ is $\mP$-balanced we need to show that, for every class  $\mP_i$ of $\mP$ and for all $x,y \in \mP_i$, the number of codes $\mC \in \mF_\mD$ containing $x$ is the same as the number of codes $\mC \in \mF_\mD$ containing $y$. We only show this for $\mP_3$.

Let $x,y \in \mP_3$ be arbitrary. Fix a basis $(z_1,...,z_t)$ of $\mD$, and let $\mB_x=(z_1,...,z_t,x,z_{t+2},...,z_N)$ and 
 $\mB_x=(z_1,...,z_t,y,z'_{t+2},...,z'_N)$ be bases of $X$. Denote by $G:X \to X$ the unique $\F_Q$-linear isomorphism that sends
$\mB_x$ to $\mB_y$. Note that $G$ preserves $\mD$ and sends $x$ to $y$. Now let $\mC \in \mF_\mD$ be an arbitrary code that contains $x$. Since $\mC \cap \mD \neq \{0\}$ by assumption and $G$ is an isomorphism, we have
$G(\mC) \cap \mD=G(\mC) \cap G(\mD)=G(\mC \cap \mD) \neq \{0\}$. Thus $G(\mC) \in \mF_\mD$. Moreover, as $x \in \mC$, we have 
$y =G(x) \in G(\mC)$. All of this shows that $G$ induces a bijection between codes in $\mF_\mD$ containing $x$ and codes in $\mF_\mD$ containing $y$.
 
We can now compute the invariants of $(\mP,\mF_\mD)$. It is immediate that $\mP_1(\mF_\mD)=|\mF_\mD|$. To compute $\mP_2(\mF_\mD)$, it suffices to use the fact that $\mF$ is $\mP$-balanced, and double count the elements of the set $\{(x,\mC) \in \mP_2 \times \mF_\mD \mid \mC \in \mF_\mD, \ x \in \mC\}$. The value of $\mP_3(\mF_\mD)$ can be obtained similarly.
\end{proof}

\section{Density Functions and Their Asymptotics}\label{sec:densityfs}

We formally define {\em density functions} and what it means for a family to be {\em sparse} or {\em dense} within a larger family.
Such notions have been used for some decades in number theory~\cite{niven}. The following definition is most apt in the context of 
our work, namely describing the asymptotic behaviour of density functions of families of error-correcting codes.

\begin{definition}
Let $S \subseteq \N$ be an infinite subset of the natural numbers. Let 
$(\mF_s \mid s \in S)$
be a sequence of finite non-empty sets indexed by $S$, and let 
$(\mF'_s \mid s \in S)$ be a sequence of sets with $\mF_s' \subseteq \mF_s$ for all
$s \in S$. The \textbf{density function} $S \to \Q$ of $\mF'_s$ in $\mF_s$ is given by
$$s \mapsto |\mF'_s|/|\mF_s|.$$

When $\lim_{s \to +\infty} |\mF'_s|/|\mF_s|$ exists and equals $\delta$, then we say that
$\mF'_s$ has \textbf{density} $\delta$ in $\mF_s$. If $\mF'_s$ has density $0$ in $\mF_s$, then 
$\mF'_s$ is \textbf{sparse} in $\mF_s$. If $\mF'_s$ has density $1$ in $\mF_s$, then 
$\mF'_s$ is \textbf{dense} in $\mF_s$.
\end{definition}

To simplify the notation, throughout the paper the variable $s$ in $\mF_s$ and $\mF'_s$ is omitted when it is clear from the context. 
We remark that notions of lower density (the $\lim_{s \to +\infty} \inf$) and upper density (the $\lim_{s \to +\infty} \sup$) are also  used and appear in the literature, but are not required here.



\subsection{Asymptotic Estimation}

Since in several instances we will obtain estimates on density functions, we recall the standard notation used to describe the asymptotic growth of functions (see~\cite{CLRS} for example).

\begin{definition}
	Let $f$ be real-valued a function defined on an infinite domain $S\subseteq \N$.
	 We denote by 
	 $\mO(f)$, $\Omega(f)$ and $\Theta(f)$ the following sets.
	 \begin{eqnarray*}
	  \mO(f) & := & 
	  \{ g:S \to \R \mid \exists \ C \in \R_{>0}  \mbox{ and }  s_0 \in S \text{ with }  0 \leq g(s) \leq C  f(s) \ \forall \  s \geq s_0 \},\\
	  \Omega(f) & := &
	  \{g:S \to \R \mid \exists \ C \in \R_{>0}  \mbox{ and }  s_0 \in S \text{ with } 0 \leq C  f(s) \leq g(s) \ \forall \ s \geq s_0\},\\
	  \Theta(f)& := &\{g:S \to \R \mid  \exists \ C_1,C_2 \in \R_{>0} \mbox{ and } s_0 \in S  \text{ with } 0 \leq C_1  f(s) \leq g(s) \leq C_2  f(s) \ \forall \ s \geq s_0\}.
	 \end{eqnarray*}

\end{definition}
 If $f,g$ are functions of more than one variable, we will put in evidence the variable with respect to which the asymptotic estimate is made by writing expressions such as
$$f \in \Theta(g) \mbox{ as } s \to +\infty,$$where all the other variables are treated as constants.
We will also need the following fact.

\begin{lemma}\label{lem:polytheta}
	Let $p(s) \in \R[s]$ be a polynomial of degree $k$. Then $p(s) \in \Theta\left(s^k\right)$ as $s \to +\infty$, when $p$ is viewed as a function $p:S \to \R$ on an infinite subset $S \subseteq \N$. 
\end{lemma}

The next simple consequences of the previous lemma will be
particularly useful in the sequel. 

\begin{lemma}\label{lem:thetas}
  Let $a\geq b$ be non-negative integers. The following hold.
    \begin{enumerate} 
  	\item\label{cor:1} The $Q$-binomial coefficient of $a$ and $b$ is a polynomial in $Q$ of degree $b(a-b)$. In particular, 
$$\qbin{a}{b}{Q} \in \Theta \left(Q^{b(a-b)} \right) \mbox{  as } Q \to +\infty.$$
  	\item\label{cor:2} For all $A_0,...,A_{b-a-1} \in \R$ we have $A_0s^{-a}+A_{1}s^{-a-1}+\cdots+s^{-b} \in \Theta\left(s^{-b} \right)$ as $s \to +\infty$. 
  \end{enumerate}
\end{lemma}

Throughout the paper, we will apply both Lemma~\ref{lem:polytheta} and Lemma~\ref{lem:thetas} without explicitly referring to them.

\subsection{Some Preliminary Formul{\ae}}

We conclude this section on density functions and asymptotic estimation by establishing some technical results that will be needed later.

\begin{proposition} \label{pro:inters}
Let $\mD \le X$ be a fixed $\F_Q$-linear code of dimension $t$. For all $0 \le k \le N$ we have
\begin{equation*}
  |\{\mC \le X \mid \dim(\mC) =k, \ \mC \cap \mD \neq \{0\}\}| = \sum_{h=1}^t \qbin{t}{h}{Q} \sum_{s=h}^t \qbin{t-h}{s-h}{Q} \qbin{N-s}{N-k}{Q} (-1)^{s-h} Q^{\binom{s-h}{2}}.
\end{equation*}
\end{proposition}

\begin{proof}
Denote by 
$(\mL, \le)$ the lattice of $\F_Q$-linear subspaces of $\mD$. For any subspace $\mH \in \mL$ define 
$f(\mH):=|\{\mC \le X \mid \dim(\mC)=k, \ \mC  \cap \mD =\mH\}|$. Then for all
$\mH \in \mL$ we have 
$$g(\mH):= \sum_{\substack{\mH' \in \mL \\ \mH' \ge \mH}}f(\mH') = 
 |\{\mC \le X \mid \dim(\mC)=k, \ \mC \ge \mH\}| = \qbin{N-h}{N-k}{Q},$$
where $h:=\dim(\mH)$.
We now use M\"{o}bius inversion~\cite[Proposition 3.7.1]{ec} in the lattice $(\mL,\le)$ and obtain, for all $\mH \in \mL$,
$$f(\mH)= \sum_{\substack{\mH' \in \mL \\ \mH' \ge \mH}} g(\mH') \ \mu(\mH,\mH'),$$
where $\mu$ denotes the M\"{o}bius function of $\mL$. Therefore for all
$\mH \in \mL$ of dimension $h$ we have
\begin{equation} \label{intermediate}
f(\mH) = \sum_{s=h}^t \sum_{\substack{\mH' \in \mL \\ \mH' \ge \mH \\ \dim(\mH')=s}} \qbin{N-s}{N-k}{Q} (-1)^{s-h} Q^{\binom{s-h}{2}} = \ \sum_{s=h}^t \qbin{t-h}{s-h}{Q} \qbin{N-s}{N-k}{Q} (-1)^{s-h} Q^{\binom{s-h}{2}}.
\end{equation}
The formula in the statement now follows from (\ref{intermediate}) and the fact that
\begin{equation*}
|\{\mC \le X  \mid \dim(\mC)=k, \ \mC \cap \mD \neq \{0\}\}| = \sum_{h=1}^t \sum_{\substack{\mH \in \mL \\ \dim(\mH)=h}} f(\mH). \qedhere
\end{equation*}
\end{proof}

Since the quantity  $|\{\mC \le X \mid \dim(\mC) =k, \ \mC \cap \mD \neq \{0\}\}|$ will arise a number of times in our results, we introduce the following notation. 
\begin{notation} \label{notaz:lambda}
	For non-negative integers $N$, $k$, $t$ and a prime power $Q$,  let
	$$\Lambda(Q;N,t,k):=\sum_{h=1}^t \qbin{t}{h}{Q} \sum_{s=h}^t \qbin{t-h}{s-h}{Q} \qbin{N-s}{N-k}{Q} (-1)^{s-h} Q^{\binom{s-h}{2}}.$$
\end{notation}

We can now give a precise asymptotic estimate of 
$\Lambda(Q;N,t,k)$ as $Q$ grows.
\begin{proposition}\label{prop:asymp}
	Let $N,k,t$ be non-negative integers satisfying $t+k -1 \leq N$. Then
	$$\lim_{Q \to +\infty} \frac{\Lambda(Q;N,t,k)}{Q^{t-1+ (k-1)(N-k)}}=1.$$
	In particular,
	$$\Lambda(Q;N,t,k)\in \Theta \left(Q^{t-1+ (k-1)(N-k)} \right) \mbox{ as } Q \to +\infty.$$
\end{proposition}

\begin{proof}
 Observe that 
 $$
\Lambda(q;N,t,k) = \sum_{h=1}^{t}\qbin{t}{h}{Q} \sum_{s=0}^{t-h}\qbin{N-s-h}{N-k}{Q}\qbin{t-h}{s}{Q}(-1)^{s}Q^{\binom{s}{2}}. 
$$
For a polynomial $p$ in indeterminate $Q$, write ${\mathrm{Lt}}(p)$ to denote its leading term.
We have  
$${\mathrm{Lt}}\left(\qbin{N-s-h}{N-k}{Q} \right) = Q^{(k-s-h)(N-k)},$$
hence
\begin{eqnarray*}
	{\mathrm{Lt}}(\Lambda(Q;N,t,k)) = {\mathrm{Lt}}\left(\sum_{h=1}^{t}\qbin{t}{h}{Q} Q^{(N-k)(k-h)}\sum_{s=0}^{t-h}\qbin{t-h}{s}{Q}(-1)^{s}Q^{\binom{s}{2}}Q^{(k-N)s}\right).
\end{eqnarray*}
The inner sum can be expressed as:
\begin{eqnarray*}
	\sum_{s=0}^{t-h}\qbin{t-h}{s}{Q}(-1)^{s}Q^{\binom{s}{2}}Q^{(k-N)s} =\prod_{s=0}^{t-h-1}\left(1-Q^{k-N+s}\right),
\end{eqnarray*}
which is in $\Theta(1)$ as $Q \to +\infty$ for each value of $h$ satisfying $t+k-h \leq N$, and contributes $1$ in the product of terms yielding the leading term of $\Lambda(Q;N,t,k)$.
Moreover, 
$${\mathrm{Lt}}\left(\qbin{t}{h}{Q} Q^{(N-k)(k-h)}\right) =Q^{h(t-h-N+k)+k(N-k)} .$$ 
It is easy to check that for $h \in \{1,...,t\}$ and $t+k-N \leq h$, the quantity $Q^{h(t+k-N-h)+k(N-k)}$ attains its maximum value at $h=1$, and hence $${\mathrm{Lt}}\left(\sum_{h=1}^{t}\qbin{t}{h}{Q} Q^{(N-k)(k-h)} \right) =Q^{t+k-N-1 + k(N-k)} =Q^{t-1+(k-1)(N-k)}.$$
The proposition follows.
\end{proof}

\section{Distance-Regular Spaces from Coding Theory}\label{sec:dist-reg}

In this section we briefly describe three distance-regular spaces in coding theory, namely, the Hamming space,
the matrix rank metric space, and the vector rank metric space. 
We also provide asymptotic estimations for some of the parameters associated with these spaces.

\begin{notation} \label{mainnotaz}
Throughout the paper, $q$ denotes a prime power, and $n$, $m$ are integers that satisfy 
$2 \le n \le m$. The results that we will obtain for the general $Q$-ary space $(X,d)$ of dimension $N \ge 2$ will be applied substituting $Q=q$, $Q=q^m$, $N=n$, $N=m$, or $N=mn$ depending on the context.
\end{notation}

\subsection{Hamming Space}
Denote by 
$\dH$ the Hamming distance on $\F_q^n$. Then $(\F_q^n,\dH)$
is a $q$-ary distance-regular space, called the \textbf{Hamming space}. The weight induced by $\dH$ is denoted by $\wH$.
We use the symbol $\rhoH$ for the Hamming-metric covering radius.
The linear codes in $\F_q^n$ are the \textbf{block codes}.
We write that $\mC$ is an $[n,k]_q$ code to say that $\mC$ is an $\F_q$-linear code in $\F_q^n$ of dimension $k$.
For all $r \in \N$, the size of the ball of radius $r$ in the Hamming space is estimated to be 
\begin{equation}\label{estimate_ball_H}
\bH(r)= \sum_{i=0}^r \binom{n}{i} (q-1)^i  \ \in 
          \Theta \left(q^r \right)  \mbox{ as } q \to +\infty.
\end{equation}
 
The following upper bounds for the covering radius and minimum distance of a linear code with the Hamming metric are well known.  They are called the \textit{redundancy bound} and the 
\textit{Singleton bound}, respectively. See \cite[Corollary 11.1.3]{hupless} and \cite[Theorem 2.4.1]{hupless} respectively.

\begin{proposition} \label{Singleton_redun}
Let $\mC \le \F_q^n$ be an $[n,k]_q$ code.
Then $\rhoH(\mC) \le n-k$ and if 
$k \ge 1$, then
$\dH(\mC) \le n-k+1$. 
\end{proposition}

An $[n,k]_q$ code $\mC \le \F_q^n$ with $k \ge 1$ and minimum distance
$\dH(\mC)=n-k+1$ is called \textbf{MDS}.
It is known that a $k$-dimensional MDS code $\mC \le \F_q^n$ exists for all $1 \le k \le n$ whenever $q \ge n-1$. Moreover, the weight distribution of an MDS $[n,k]_q$ code is uniquely determined. See for example~\cite[Chapter 11]{macws}.

\begin{remark} \label{MDSbalanced}
For all $1 \le k \le n$, the family of $k$-dimensional MDS codes
$\mC \le \F_q^n$ is $\mP(\wH)$-balanced. This is easy to see. If $x,y \in \F_{q}^n$ have the same Hamming weight, then there is a monomial
transformation $\sigma:\F_q^n \to \F_q^n$ taking $x$ to $y$. This induces a bijection on the $\fq$-linear codes containing $x$ and those containing $y$.
In particular, $\sigma$, being a Hamming distance isometry, maps an MDS $[n,k]_q$ code $\mC$ containing $x$ to the equivalent MDS code $\sigma(\mC)$ containing $y$.
\end{remark}

We conclude this subsection on the Hamming space 
by giving explicit formul{\ae}  for its intersection numbers.
These expressions are well-known (see \cite[Chapter 21]{macws} for $q=2$). 
We also include asymptotic estimations for such numbers.

Let $i,j,k$ be integers satisfying   $0 \le i,j,k \le n$.
	The intersection numbers of the Hamming space  $\F_{q}^n$ are given by 
	\begin{eqnarray}\label{eqpijkham}
		\pH{i}{j}{k} & = & \sum_{r\geq 0}\binom{i}{r}\binom{n-i}{n-k-r}\binom{i-r}{j+i-k-2r}(q-1)^{k-i+r}(q-2)^{j+i-k-2r}\\
		& = & \sum_{r\geq 0}^{\min\{i,j,n-k,(j+i-k)/2\}}\theta(n,i,j,k,r)(q^{j-r} + \cdots), \nonumber
	\end{eqnarray}
	where 
	\begin{equation} \label{exptheta}
	\theta(n,i,j,k,r):=\binom{i}{r}\binom{n-i}{n-k-r}\binom{i-r}{j+i-k-2r}=\binom{i}{r}\binom{n-i}{k-i+r}\binom{i-r}{k-j+r}.
	\end{equation}

\begin{lemma}\label{lem:thetapos}
	Let $0 \le r, i,j,k \le n$ be integers. Then the expression 
	$\theta(n,i,j,k,r)$ in (\ref{exptheta}) is positive if and only if the following three inequalities hold:
	$$0 \leq r \leq i, \qquad i-k \leq r \leq n-k, \qquad j-k \leq r \leq (i+j-k)/2.$$
In particular $\pH{i}{j}{k} = 0 $ if $i > j+k$ or $j>i+k$ or $k > i+j$.
\end{lemma}

\begin{proof}
	The fact that $\theta(n,i,j,k,r)$ is positive if and only if $r$ satisfies the inequalities shown can be seen by inspection of its binomial factors.
	The value $\pH{i}{j}{k}$ is zero if and only if $\theta(n,i,j,k,r)=0$ for all $r$ satisfying $0 \leq r \leq i,j,n-k,(j+i-k)/2$. It can be checked that this occurs if $i > j+k$ or $j>i+k$ or $k > i+j$.	
\end{proof}

We now derive asymptotic estimates of the intersection numbers of the Hamming space as the field size $q$ grows.

\begin{proposition}\label{lem:pijkasym}
	Let $0 \le i,j,k \le n$ be integers such that intersection number
	 $\pH{i}{j}{k}$ for the Hamming space $\F_{q}^n$ is positive. 
	Let $r_0 = \min\{r \in \N \mid \theta(n,i,j,k,r)>0 \}$. Then
	$$\lim_{q \to +\infty} \pH{i}{j}{k} = \theta(n,i,j,k,r_0)q^{j-r_0}. $$
	In particular,
	as $q \to +\infty$,
	$$\pH{i}{j}{k} \in \left\{
	\begin{array}{ll}
	\Theta \left(q^j\right) & \text{ if } 0 \leq i,j \leq k \leq i+j, \\
	\Theta \left(q^{j-i+k}\right)& \text{ if } 0\leq j,k\leq i\leq j+k , \\
	\Theta \left(q^k\right) & \text{ if } 0\leq i,k \leq j \leq i+k.\\
	\end{array}
	\right.$$
\end{proposition}

\begin{proof}
	If there exists some non-negative $r \leq i,n-k,(i+j-k)/2$ such that $\theta(n,i,j,k,r)>0$, then $\pH{i}{j}{k}>0$, and  
	the leading term in $q$ of 
	$$\pH{i}{j}{k} = \sum_{r\geq 0}^{\min\{i,j,n-k,(j+i-k)/2\}}\theta(n,i,j,k,r) (q-1)^{k-i+r}(q-2)^{j+i-k-2r}$$
	is $\theta(n,i,j,k,r_0)q^{j-r_0}$.
	Let $r=\max\{ 0,i-k,j-k\}$ and suppose that $\pH{i}{j}{k}$ is non-zero. Then from Lemma~\ref{lem:thetapos} it holds
	that $k\leq i+j,j \leq i+k, i \leq j+k$. 
    We claim that $r_0 =r$. Clearly $r_0 \leq r$ since otherwise at least one of inequalities in the statement of Lemma~\ref{lem:thetapos} does not hold.
    If $r=0$ then $i,j \leq k$ and $\theta(n,i,j,k,0)>0$ only if $k\leq i+j$, so $r=r_0$ and $j-r_0=j$.
    If $r=i-k$ then $j,k \leq i$ and $\theta(n,i,j,k,i-k)>0$ only if $i\leq j+k$, so $r=r_0$ and $j-r_0 = j-i+k$.
    If $r=j-k$ then $i,k \leq j$ and $\theta(n,i,j,k,j-k)>0$ only if $j\leq i+k$, so $r=r_0$ and $j-r_0 = k$.
\end{proof}

\subsection{Matrix Rank-Metric Space}
The \textbf{rank distance} between matrices $x,y \in \F_q^{n \times m}$ is defined to be
$\drk(x,y):=\rk(x-y)$. Then $(\F_q^{n \times m},\drk)$ is a $q$-ary distance-regular space, called the (\textbf{matrix}) \textbf{rank metric} \textbf{space}. We use the symbol $\rhork$ for the rank metric covering radius.
A code in $\F_q^{n \times m}$ is called a (\textbf{matrix}) \textbf{rank metric code}.  
Clearly, the weight induced by 
$\drk$ is matrix rank.
The size of the ball of radius $r \in \N$ in $\mat$ is given by
\begin{equation} \label{estimate_ball_rk}
\brk(r) = \sum_{i=0}^{r} \qbin{n}{i}{q}\prod_{j=0}^{i-1}(q^m-q^j) \ \in \left\{ \begin{array}{ll}
          \Theta \left(q^{r(n+m-r)}\right) & \mbox{ as } q \to +\infty, \\
          \Theta \left(q^{rm}\right)   & \mbox{ as } m \to +\infty.
\end{array}\right.
\end{equation}

There exists a \textit{Singleton-type bound} also for rank metric codes. See \cite[Theorem 5.4]{del} for a proof using association schemes, or \cite[Section 3]{costch} for a linear algebra proof.

\begin{proposition} \label{rankSing}
The dimension of a matrix code in $\mat$ with minimum rank distance $d$ is at most $m(n-d+1)$.
\end{proposition}

A rank metric code $\mC \le \F_q^{n \times m}$  of minimum rank distance $d$ and dimension
$\dim(\mC)=m(n-d+1)$ is called a {\bf maximum rank distance} code and we say that $\mC$ is an \textbf{MRD} code. 
It is known (see~\cite[Section 6]{del}) that for every $1 \le d \le n$ there exists an MRD code $\mC \le \F_q^{n \times m}$ of  minimum distance  $d$. Notice that we always assume $m \ge n$. See Notation~\ref{mainnotaz}.

\begin{remark}
For all $1 \le d \le n$, the family of MRD codes
$\mC \le \F_q^{n \times m}$ of minimum rank distance $d$ is $\mP(\rk)$-invariant.
This can be seen by a very similar argument as given in Remark~\ref{MDSbalanced}. 
If $x,y \in \F_q^{n \times m}$ have the same rank over $\fq$, then there exist invertible $\F_{q}$-matrices $A,B$
satisfying $AxB =y$. Then multiplication by $A$ and $B$ induces a bijection between the $\fq$-linear codes containing $x$ and those containing $y$,
which moreover preserves the MRD property, being an $\fq$-linear isometry of $\fq^{m \times n}$.   
\end{remark}

\subsection{Vector Rank-Metric Space}

Define the $\textbf{rank weight}$ $\wrk(x)$ of a vector $x \in \F_{q^m}^n$ as the dimension over $\F_q$ of the subspace generated by its components. The \textbf{rank distance}
between vectors 
$x,y \in \F_{q^m}^n$ is $\drk(x-y)$. Then the pair
$(\F_{q^m}^n,\drk)$ is a $q^m$-ary distance regular space, called the (\textbf{vector}) \textbf{rank metric} \textbf{space}
(see~\cite{gabid} for further details).
The codes in $\F_{q^m}^n$ are the \textbf{vector rank metric codes}. 
We write that $\mC$ is an $[n,k]_{q^m}$ code to say that $\mC$ is an $\F_{q^m}$-linear code in $\F_{q^m}^n$ of dimension $k$ over $\F_{q^m}$. The rank metric covering radius of such a code is denoted by $\rhork(\mC)$.

Clearly, $(\F_{q^m},\drk)$ can be also viewed as a $q$-ary distance-regular space, which in fact is isomorphic (as $\F_q$-linear space) and isometric to the matrix space $(\mat,\drk)$. Let us make these isomorphisms more explicit. Given an $\F_{q}$-basis $\Gamma=\{\gamma_1,...,\gamma_m\}$ of $\F_{q^m}$ and given a vector $x \in \F_{q^m}^n$, denote by $\Gamma(x)$ the $n \times m$ matrix over $\F_q$ defined by
$$x_i=\sum_{j=1}^m \Gamma_{ij}(x) \gamma_j \quad \mbox{for all } 1 \le i \le n.$$
Then the following hold (see e.g. \cite[Section 1]{costch}).

\begin{proposition} \label{isome}
For every $\F_{q}$-basis $\Gamma$ of $\F_{q^m}$, the map $x \mapsto \Gamma(x)$ is an $\F_q$-linear bijective isometry
 from $(\F_{q^m},\drk)$ to $(\mat,\drk)$. In particular,
 if $\mC \le \F_{q^m}^n$ is an $[n,k]_{q^m}$ vector rank metric code, then $\Gamma(\mC)$ is an $\F_q$-linear matrix rank metric code of dimension $mk$ over $\F_q$ having the same minimum distance, weight distribution and covering radius as $\mC$.
\end{proposition} 

Proposition~\ref{isome} also implies that the ball of given radius, say $r$, in $(\F_{q^m}^n,\drk)$ has the same cardinality as the ball of radius $r$ in  
$(\mat,\drk)$. Such cardinality is denoted by $\brk(r)$ in both cases, and the estimates in~\ref{estimate_ball_rk} remain valid in the context of vector rank metric codes.

A \textit{redundancy bound} and a \textit{Singleton-type bound} are known for vector rank metric codes. 

\begin{proposition}\label{lem:rkred}
Let $\mC \le \F_{q^m}^n$ be an $[n,k]_{q^m}$ code.
We have $\rhork(\mC) \le n-k$. Moreover, if 
$k \ge 1$, then
$\drk(\mC) \le n-k+1$. 
\end{proposition} 
\begin{proof}
To see that $\rhork(\mC) \le n-k$, it suffices to observe that $\rhork(\mC) \le \rhoH(\mC)$ and apply Proposition~\ref{Singleton_redun}. See~\cite{gabid} for the Singleton-type bound,
or combine Propositions~\ref{rankSing} and~\ref{isome}.
\end{proof}

A vector rank metric code $\mC \le \F_{q^m}^n$  of minimum rank distance $d$ and dimension
$\dim(\mC)=n-d+1$ over $\F_{q^m}$ is called \textbf{MRD}.

\begin{remark}
	With respect to the vector rank metric codes that arise in network coding~\cite{kk}, such as the Delsarte-Gabidulin codes~\cite{del,gabid}, the parameter $m$ is the effective {\em block length} of the code, representing the capacity of the network edges.
	Therefore it is natural to examine the asymptotic behavour of families of vector rank metric codes as $m$ grows.
\end{remark}

\section{Density of Codes of Fixed Minimum Distance}\label{sec:densmindist}

In this section we present the first of our main results, and 
give an asymptotic estimate for the 
number of non-MDS and non-MRD codes as $q \to +\infty$ and $m\to +\infty$, respectively. As simple corollaries, we recover the result that MDS and vector rank metric MRD codes are dense among all codes with the same length and dimension. See in particular~\cite{NTRR}.

\begin{theorem} \label{th_dist1}
Let $0 \le t < N$ be an integer, and let $\mD \le X$ be a $t$-dimensional linear code with $d(\mC) \geq 2$. Let $k$ and $d$ be any integers that satisfy $t\le k \le N$ and $2 \le d \le d(\mD)$. Define the families of codes $\mF:=\{\mC \le X \mid \mD \le \mC, 
\ \dim(\mC)=k\}$ and $\mF':=\{\mC \le X \mid \mD \le \mC, \ \dim(\mC)=k, \ d(\mC)<d\}$. We have 
$$|\mF'|/|\mF| \le \frac{Q^k-Q^t}{(Q-1)(Q^N-Q^t)}  \left( \ball(d-1) -1\right).$$
In particular, for all integers $k$ and $d$ with $0 \le k \le N$ and $d \ge 2$, there exist at most 
$$\frac{Q^k-1}{(Q-1)(Q^N-1)}  \qbin{N}{k}{Q} \left( \ball(d-1) -1\right).$$
linear codes $\mC \le X$ of dimension $k$ over $\F_Q$ and minimum distance $<d$.
\end{theorem}

\begin{proof}
We will apply Lemma~\ref{KEY} to the family 
$\mF=\{\mC \le X \mid \mD \le \mC, \ \dim(\mC)=k\}$. Let
$\mP = \{\mD,X \backslash \mD \}$, which is the partition of Proposition~\ref{examp1}.
Let $f: X \to \R$ the characteristic function of the set $\Ball(0,d-1) \setminus \{0\} \subseteq X$. Since $d(\mD) \ge d$ by assumption, one obtains
\begin{equation} \label{e1}
\sum_{\mC \in \mF} |\mC \cap (\Ball(0,d-1) \setminus \{0\})|
= \frac{|\mF| \cdot (Q^k-Q^t)}{Q^N-Q^t} (\ball(d-1)-1).
\end{equation}
Let $\mF':=\{\mC \in \mF \mid d(\mC) < d\} \subseteq \mF$. 
Then
\begin{equation} \label{e2}
\sum_{\mC \in \mF} |\mC \cap (\Ball(0,d-1) \setminus \{0\})|
= \sum_{\mC \in \mF'} |\mC \cap (\Ball(0,d-1) \setminus \{0\})|
 \ge |\mF'| \cdot (Q-1),
\end{equation}
since $\omega(\lambda x) = \omega(x)$ for all non-zero $\lambda \in \F_Q$ (see Definition~\ref{defregular}). The result now follows, combining (\ref{e1}) and (\ref{e2}).
\end{proof}

We now apply Proposition~\ref{pro:inters} and Theorem~\ref{th_dist1}  to codes endowed with the Hamming and the rank metric, taking as $\mD$ the zero code $\{0\} \le X$. We give  asymptotic estimates of the density function of linear codes whose minimum distance is upper bounded.

\begin{corollary} \label{coro_dist1}
Let $k$ and $d$ be integers with $1 \le k \le n$, $2 \le d \le n$, and $d+k \le n+2$.
Denote by $\mF$ the family of $[n,k]_q$ codes, and let
$\mF':=\{\mC \in \mF \mid \dH(\mC) < d\}$.
Then
$$
\qbin{n}{k}{q}^{-1} \Lambda(q;n,d-1,k) \le \DD{\mF'}{\mF} \le \frac{q^k-1}{(q-1)(q^n-1)}  (\bH(d-1)-1).
$$
In particular,
$$\DD{\mF'}{\mF} \in \Theta \left(q^{d+k-2-n} \right) \mbox{ as } q \to +\infty.$$
\end{corollary}

\begin{proof}
    Let $\mD \le \F_q^n$ be the space generated by the first 
    $d-1$ standard basis vectors in $\F_q^n$, i.e., $\mD = \F_q^{d-1} \oplus 0^{n-d+1}$. 
    Then any code $\mC$ that intersects $\mD$ non-trivially has minimum distance
    at most $d-1$. From Proposition~\ref{pro:inters}, the number of such $k$-dimensional codes is $\Lambda(q;n,d-1,k)$.
    Therefore, the number of $[n,k]_q$ codes of minimum Hamming distance at most $d-1$ is at least $\Lambda(q;n,d-1,k)$.
The upper bound on $\DD{\mF'}{\mF}$ is immediate from Theorem~\ref{th_dist1}.

We now show the second part of the statement. 
        Since $d+k-2\leq n$, we may apply Proposition~\ref{prop:asymp} to see that $\Lambda(q;n,d-1,k) \in \Theta\left(q^{d+k-2-n}\right)$. Combining this with the asymptotic estimate of the Hamming sphere given in (\ref{estimate_ball_H}) we conclude that 
\begin{equation*}
\DD{\mF'}{\mF} \in \Omega\left(q^{d+k-2-n}\right) \cap {\cal O}\left(q^{d+k-2-n}\right) \subseteq \Theta\left(q^{d+k-2-n}\right) \mbox{ as } q \to +\infty. \qedhere
\end{equation*}
\end{proof}

As a direct consequence of Corollary~\ref{coro_dist1}, we deduce that the MDS codes are dense among the linear codes of fixed dimension.
Moreover, we give a precise asymptotic estimate of the density function of the non-MDS codes in $\F_q^n$ as $q$ grows.

\begin{corollary}
Let $k$ be an integer with $1 \le k \le n$, and let $\mF$ be the family on $[n,k]_q$ codes. Denote by 
		$\mF'$ the family of $[n,k]_q$ codes that are not MDS.
Then $$\DD{\mF'}{\mF} \in \Theta\left( q^{-1}\right) \mbox{ as } q \to +\infty.$$ 
In particular, the probability 
that a randomly chosen $[n,k]_q$ code is MDS approaches 1 as $q \to +\infty.$
\end{corollary}

We now consider vector rank metric codes.

\begin{corollary} \label{coro_dist2}
Let $k$ and $d$ be integers with $1 \le k \le n$, 
$2 \le d \le n$ and $d+k-2 \leq n$.
Let $\mF$ be the family of $[n,k]_{q^m}$ vector rank metric codes, and let
$\mF':=\{\mC \in \mF \mid \drk(\mC) < d\}$. Then 
$$\qbin{n}{k}{q^m}^{-1} \Lambda(q^m;n,d-1,k) \le \DD{\mF'}{\mF} \le \frac{q^{mk}-1}{(q^m-1)(q^{mn}-1)}
 (\brk(d-1) -1).$$
In particular,
$$\DD{\mF'}{\mF} \in \Theta \left(q^{m(d+k-2-n)}\right) \mbox{ as } m \to +\infty.$$
\end{corollary}

\begin{proof}
	Since the $\F_{q}$-rank of a vector over $\F_{q^m}$ is at most its Hamming weight, from Corollary~\ref{coro_dist1} we have that
	$$\DD{\mF'}{\mF} \geq \qbin{n}{k}{q^m}^{-1} \Lambda(q^m;n,d-1,k) \in \Omega\left(q^{m(d-(n-k+1)-1)}\right)\mbox{ as } m \to \infty,$$
	where the last asymptotic estimate follows from Proposition~\ref{prop:asymp} and the fact that
	$q^m \to +\infty$ as $m \to +\infty$.
	The upper bound on $\DD{\mF'}{\mF}$ follows directly from Theorem~\ref{th_dist1}.
	As before, an estimate of such upper bound is follows now using the asymptotic estimate of the rank metric ball given in~(\ref{estimate_ball_rk}).
\end{proof}

In analogy with the Hamming metric, Corollary~\ref{coro_dist2} shows that the MRD codes are dense in the family of vector $\F_{q^m}$-linear codes. This fact was first shown in~\cite{NTRR}. Note that Corollary~\ref{coro_dist2} also allows us to give a precise asymptotic estimate of the density function of the non-MRD vector codes in $\F_{q^m}^n$ as $m$ grows.

\begin{corollary}
Let $k$ be an integer with  $1 \le k \le n$, and let $\mF$ be the family of $[n,k]_{q^m}$ codes. Denote by 
		$\mF'$ the family of $[n,k]_{q^m}$ codes that are not MRD.
Then $$\DD{\mF'}{\mF} \in \Theta\left(q^{-m}\right) \mbox{ as } m \to \infty.$$ 
In particular, the probability 
that a randomly chosen $[n,k]_{q^m}$ code is MRD approaches 1 as $m$ approaches $+\infty$.
\end{corollary}

We now apply the same machinery to study matrix rank metric codes.

\begin{corollary} \label{coro_dist333}
Let $k$ and $d$ be integers with $1 \le k \le mn$ and 
$2 \le d \le n$.
Let $\mF$ be the family of $k$-dimensional matrix rank metric codes in $\mat$, and let
$\mF':=\{\mC \in \mF \mid \drk(\mC) < d\}$. Then 
$$\DD{\mF'}{\mF} \le \frac{q^{k}-1}{(q-1)(q^{mn}-1)}
 (\brk(d-1) -1).$$
In particular,
\begin{equation} \label{exxp}
\DD{\mF'}{\mF} \in \mO \left(q^{k-1-mn+(d-1)(m+n-d+1)}\right) \mbox{ as } q \to +\infty.
\end{equation}
\end{corollary}

\begin{proof}
Again, the result follows combining Theorem~\ref{th_dist1} with the asymptotic estimate of the rank metric ball given in~(\ref{estimate_ball_rk}).
\end{proof}

Following the notation of Corollary~\ref{coro_dist333}, it is easy to derive conditions on $(m,n,k,d)$ such that the exponent of $q$ in the right-hand side of (\ref{exxp}) is negative. For those values of $(m,n,k,d)$ Corollary~\ref{coro_dist333} shows the sparsity of the family of matrix codes 
in $\mat$ of dimension $k$ and minimum distance $<d$ within the family of $k$-dimensional codes.

\begin{remark}
	The method used in this section cannot be applied to establish the density of $\fq$-linear MRD codes as $q \to +\infty$ or $m \to +\infty$. In fact, as we show in the next section, this is false; the $\fq$-linear MRD
	codes of minimum distance $d$ are \textit{never} dense in the family of $m \times n$ matrix codes of dimension $m(n-d+1)$, both as $q \to +\infty$ and as $m \to +\infty$.
	
	Note moreover that the direct analogues for matrix codes of the lower bounds in Corollaries~\ref{coro_dist1} and~\ref{coro_dist2} give a lower bound on the proportion of non-MRD codes that is far from sharp,
	resulting in the estimate that the density function of non-MRD codes being in $\Omega(q^{-1})$, which approaches zero as $q$ grows.	
\end{remark}

In~\cite{NTRR}, the authors give a lower bound on the density of MRD $[n,k]_{q^m}$ codes as follows.
Fix $k$ with $1 \le k \le n$, and let $\mF$ be the 
family of $[n,k]_{q^m}$ vector rank metric codes.
Denote by $\overline{\mF}$ the family of $[n,k]_{q^m}$ vector rank metric codes whose generator matrix (in reduced row-echelon form) has an identity in the first $k$ columns. Finally, let 
$\mF''$ be the family of $[n,k]_{q^m}$ MRD codes. It is well known that $\mF'' \subseteq \overline{\mF}$. In \cite[Theorem 26]{NTRR} it is shown that
$$\DD{\mF''}{\overline{\mF}} \ge 1- \sum_{r=0}^k \qbin{k}{k-r}{q} \qbin{n-k}{r}{q} q^{r^2} q^{-m},$$ 
i.e., $$\DD{\mF''}{\mF} \ge \DD{\overline{\mF}}{\mF} \left(1- \sum_{r=0}^k \qbin{k}{k-r}{q} \qbin{n-k}{r}{q} q^{r^2} q^{-m}\right).$$
We can therefore re-state the bound as follows.

\begin{theorem}[Theorem 26 of~\cite{NTRR}]
Let $k$ be an integer with $1 \le k \le n$. Let $\mF$ be the family of $[n,k]_{q^m}$ codes, and let 
$\mF''$ be the family of $[n,k]_{q^m}$ codes that are MRD.
We have
 $$\DD{\mF''}{\mF} \ge q^{mk(n-k)} \qbin{n}{k}{q^m}^{-1}\left(1- \sum_{r=0}^k \qbin{k}{k-r}{q} \qbin{n-k}{r}{q} q^{r^2} q^{-m}\right).$$
\end{theorem}

   It can be checked that the bound of Corollary~\ref{coro_dist2} is tighter than that of \cite[Theorem~26]{NTRR} for all parameters. Both formul{\ae} rapidly approach 1 as $m$ increases.

\begin{remark}
Fix an integer $k$ with $1 \le k \le n$. The proof of~\cite{NTRR} on the density of $k$-dimensional $\F_{q^m}$-linear MRD codes uses the fact that these cardinality-optimal codes
correspond to the non-zeroes of a polynomial in $k(n-k)$ variables with coefficients in $\F_{q^m}$.
 In the argument of~\cite{NTRR}, it is crucial that the degree of $p$ satisfies
\begin{equation} \label{limkey}
\lim_{m \to +\infty} \deg(p)/q^m=0,
\end{equation}
from which the density of $\F_{q^m}$-linear MRD codes as $m \to +\infty$ can be deduced by the well-known Schwartz-Zippel Lemma~\cite{schwartz,zippel}.

We notice that the proof technique illustrated above cannot be applied to deduce that $\F_q$-linear matrix MRD codes are dense as $q \to +\infty$. These optimal codes of dimension $m(k-d+1)$ over $\F_q$, where $d$ denotes the minimum distance, can easily be described 
as the non-zeroes of a multivariate polynomial in $m^2(n-d+1)(d-1)$ variables and coefficients in $\F_q$. However, the degree of any such  polynomial, say $p$, does \textit{not} satisfy
\begin{equation} \label{limitno}
\lim_{q\to+\infty} \deg(p)/q=0
\end{equation} in general, thereby preventing the application of the Schwartz-Zippel lemma as $q \to +\infty$.
The fact that the equality in (\ref{limitno}) does not hold will follow \textit{a posteriori} from the non-density of $\F_q$-linear MRD codes, which we establish in the next section of the paper.
\end{remark}

\section{$\F_q$-Linear MRD Codes Are not Dense}\label{sec:mrdnotsense}

We now obtain a lower bound on the density function of the non-MRD codes in the family of $\fq$-linear
codes of a fixed dimension. 
 We show that, far from being sparse in this set, this family asymptotically forms at least half of all 
$\fq$-linear codes of the same dimension, as both $q$ and $m$ approach infinity. 
This fact is in sharp contrast with the asymptotic behaviour of $\F_q$-linear MDS codes and $\F_{q^m}$-linear MRD codes.

Note that we only examine 
dimensions that are a multiple of $m$, as for other dimensions MRD codes do not exist (see Section~\ref{sec:dist-reg}).

\begin{theorem} \label{bigthm}
Fix an integer $d$ with $2 \le d \le n$, and let $k=m(n-d+1)$. Let $\mF'$ be the family of $k$-dimensional non-MRD codes in
$\F_{q}^{m \times n}$, that is, 
$\mF'=\{\mC \le \mat \mid \dim(\mC)=k,  \ \drk(\mC) \leq n-d\}$.
We have
\begin{equation*} 
|\mF'| \ge
q \cdot \Lambda(q;mn,m(d-1),k) \cdot \left( 1- 
\frac{\left(q^k-1\right)\left(q^{mn-k}-1\right)}{2\left(q^{mn}-q^{mn-k} \right)}\right).
\end{equation*}
\end{theorem}

\begin{proof}
To simplify the notation throughout the proof, define $\mD^\times:=\mD \setminus \{0\}$
for any linear code $\mD \le \mat$. 
The argument contains multiple claims, which we prove separately.
We start with the following simple observation.

\begin{claim} \label{claimA}
Without loss of generality, we can assume $n \ge 2(d-1)$.
\end{claim}
\begin{clproof}
Trace-duality in $\mat$ gives a bijection between the MRD codes of distance $d$ and dimension $k$ and the MRD codes of distance $n-d+2$ and dimension $mn-k$. See e.g. \cite[Section 4]{alb1}. Now observe that if $n < 2(d-1)$, then  $n>2((n-d+2)-1)$ and $2 \le n-d+2 \le n$. 
\end{clproof}

We henceforth assume that $n \ge 2(d-1)$. 
Divide $n$ by $d-1$ with remainder and write $n=a(d-1)+b$ with $a,b \in \Z$ and $0 \le b \le d-2$. Fix \textit{any} integer $M$ with 
$$2 \le M \le \frac{q^n-q^b}{q^{d-1}-1}-q^b+1.$$ Since $n \ge 2(d-1)$, we have
\begin{equation} \label{possq}
\frac{q^n-q^b}{q^{d-1}-1}-q^b+1 \ge q.
\end{equation}
Let $U_1,...,U_M \le \F_q^n$ be $M$ subspaces 
of $\F_q^n$ of dimension $d-1$ with the property that $U_i \cap U_j=\{0\}$ for all
$i,j \in \{1,...,M\}$ with $i \neq j$. The fact that $M$ spaces with this property exist follows from known results 
on \textit{partial spreads} in finite geometry.
We refer the reader to~\cite{beutel} or \cite[Section 2]{psgorlaravagnani} for further details.

For all $1 \le i \le M$ define  
$\mD_i:=\{x \in \mat \mid \mbox{colsp}(x) \le U_i\} \le \mat$,
where $\mbox{colsp}(x) \le \F_q^n$ denotes the $\F_q$-linear space generated by the columns of 
$x$.  The properties of the $\mD_i$'s can be summarized as follows.

\begin{claim} \label{claimB}
For all $1 \le i \le M$ the vector space $\mD_i \le \mat$ has dimension $m(d-1)$ over $\F_q$, and every
$x \in \mD_i$ has $\rk(x) \le d-1$. Moreover,
$\mD_i \cap \mD_j =\{0\}$ for all $i,j \in \{1,...,M\}$ with $i \neq j$.
\end{claim}
\begin{clproof}
The dimension of each $\mD_i$ is computed in \cite[Lemma 26]{alb1}.
The second property  follows directly from the definition of rank.
Finally, observe that if $i,j \in \{1,...,M\}$, then we have 
$\mD_i \cap \mD_j =\{x \in \mat \mid \mbox{colsp}(x) \le U_i \cap U_j\}$.
Thus if $i \neq j$ we have $\mD_i \cap \mD_j=\{0\}$.
\end{clproof}

For all $1 \le i \le M$, define the family $\mF_i:=\{\mC \in \mF \mid \mC\cap \mD_i^\times \neq \emptyset \} \subseteq \mF$. 
The following hold.
\begin{claim} \label{claimC}
Let $1 \le i \le M$ and $x \in \mat \setminus \mD_i$. Then
$$|\{ \mC \in \mF_i \mid x \in \mC \}| \le |\mF_i| \cdot \frac{q^k-1}{q^{mn}-q^{m(d-1)}}.$$
\end{claim}
\begin{clproof}
Fix $i$, and let $\mP$ be the partition of size 3 given by
$\mP_1=\{0\}$, $\mP_2=\mD_i^\times$ and $\mP_3=\mat \setminus \mD_i$.
It was shown in Proposition~\ref{examp3} that the family $\mF_i$ is $\mP$-balanced. Furthermore, it holds that
 \begin{equation*}
 |\{\mC \in \mF_i \mid x \in \mC \}| = \mP_3(\mF_i)
 =\frac{\sum_{\mC \in \mF_i}|\mC \cap (\mat \setminus \mD_i)|}{q^{mn}-q^{m(d-1)}} \le |\mF_i| \cdot \frac{q^k-1}{q^{mn}-q^{m(d-1)}}.
\qedhere
 \end{equation*}
\end{clproof}

Construct $M$ subsets of $\mat$ as follows:
$$A_1:=\emptyset, \qquad \quad A_i := \bigcup_{1 \le j <i} \mD_j^\times \quad \mbox{ for }
2 \le i \le M.$$
Since $A_1=\emptyset$ and the $\mD_i$'s are pairwise disjoint by Claim~\ref{claimB}, we have
\begin{equation} \label{cardAi}
|A_i| = (i-1)\left( q^{m(d-1)} -1\right) \quad \mbox{for all } 1 \le i \le M.
\end{equation}
We now introduce the auxiliary families of codes 
$ \overline{\mF}_i:=\{\mC \in \mF \mid \mC \cap \mD^\times_i \neq \emptyset \mbox{ and } \mC \cap A_i = \emptyset\} \subseteq \mF_i$, for $1 \le i \le M$. Observe that
\begin{equation}|\mF'| \ge \sum_{i=1}^M |\overline{\mF}_i|.
\label{ineqsum}
\end{equation}
The size of each
$\overline{\mF}_i$ can be lower bounded as follows. 

\begin{claim} \label{claimD}
For all $1 \le i \le M$ we have
\begin{equation*}
|\overline{\mF}_i| \ge |\mF_i| \left( 1- (q-1)^{-1} \frac{q^k-1}{q^{mn}-q^{m(d-1)}} \cdot |A_i|\right).
\end{equation*}
\end{claim}

\begin{clproof}
The result is immediate for $i=1$, as $A_1=\emptyset$. Now fix an arbitrary index $i \in \{2,...,M\}$, and define
$\hat{\mF}_i:=\{\mC \in \mF \mid \mC \cap \mD_i^\times \neq \emptyset \mbox{ and } \mC \cap A_i \neq \emptyset\}$.
Since $A_i$ is the union of linear spaces (without the zero vector), we have
$$|\hat{\mF}_i|\cdot (q-1) \le \sum_{\mC \in \mF_i} |\mC \cap A_i|,$$ from which we obtain
\begin{equation} \label{wanted}
|\overline{\mF}_i| =|\mF_i|-|\hat{\mF}_i| \ge |\mF_i|- (q-1)^{-1} \sum_{\mC \in \mF_i} |\mC \cap A_i|.
\end{equation}
On the other hand, by definition,
\begin{equation} \label{good}
\sum_{\mC \in \mF_i} |\mC \cap A_i| = \sum_{x \in A_i} |\{\mC \in \mF_i \mid x \in \mC\}|.
\end{equation}
By Claim~\ref{claimB}, we have $A_i \subseteq \mat \setminus \mD_i$. Therefore by Claim~\ref{claimC} and (\ref{good}) we conclude
\begin{equation}\sum_{\mC \in \mF_i} |\mC \cap A_i| \le |\mF_i| \cdot \frac{q^k-1}{q^{mn}-q^{m(d-1)}} \cdot |A_i|.  \label{crucialeq}
\end{equation}
The claim now follows combining (\ref{wanted}) and (\ref{crucialeq}). 
\end{clproof}

Note that by Proposition~\ref{pro:inters} and Notation~\ref{notaz:lambda} we have
$\Lambda:=\Lambda(q;mn,mn-k,k) = |\mF_i|$ for all $1 \le i \le M$.
Therefore combining (\ref{ineqsum}), Claim~\ref{claimD} and (\ref{cardAi}) we obtain
\begin{eqnarray}\label{eq:ineq1}
|\mF'| &\ge& \sum_{i=1}^M \Lambda \cdot \left( 1- (q-1)^{-1} \frac{q^k-1}{q^{mn}-q^{mn-k}} \cdot |A_i|\right) \nonumber \\
&=& \Lambda \cdot \left( M - \frac{\left(q^k-1\right)\left(q^{mn-k}-1\right)}{(q-1)\left(q^{mn}-q^{mn-k} \right)} \frac{M(M-1)}{2}\right).
\label{finOK}
\end{eqnarray}

Since the inequality in (\ref{possq}) holds, we can now set $M=q$ (we will comment later in Remark~\ref{rem:commM} on this choice of $M$) 
and obtain 
\begin{equation} \label{concl1tolimit}
|\mF'| \ge
q \cdot \Lambda \cdot \left( 1- 
\frac{\left(q^k-1\right)\left(q^{mn-k}-1\right)}{2\left(q^{mn}-q^{mn-k} \right)}\right),
\end{equation}
as desired.
\end{proof}

We now compute the limit as $q\to+\infty$ of the lower bound of the density function given in Theorem~\ref{bigthm}. 

\begin{corollary}
Fix an integer $d$ with $2 \le d \le n$, and set $k:=m(n-d+1)$. Define the families
$\mF:=\{\mC \le \mat \mid \dim(\mC)=k\}$ and $\mF':=\{\mC \in \mF \mid \mC \mbox{ is not MRD}\}$.
Then for every $\varepsilon \in \R_{>0}$ there exists $q_\varepsilon \in \N$ such that
$$\DD{\mF'}{\mF} \ge \frac{1}{2}-\varepsilon$$ for all prime powers $q \ge q_\varepsilon$.
In particular, 
$\lim_{q \to \infty} \DD{\mF'}{\mF} \ge 1/2$, provided the limit exists.
\end{corollary}

\begin{proof}
By Theorem~\ref{bigthm} we have 
\begin{equation} \label{concl1}
\frac{|\mF'|}{|\mF|} \ge
\frac{q \cdot  \Lambda(q;mn,m(d-1),k)}{|\mF|} \left( 1- 
\frac{\left(q^k-1\right)\left(q^{mn-k}-1\right)}{2\left(q^{mn}-q^{mn-k} \right)}\right). 
\end{equation}
Proposition~\ref{prop:asymp} allows us to compute
\begin{equation*} \label{thelimit}
\lim_{q\to +\infty} \frac{q\cdot  \Lambda(q;mn,m(d-1),k)}{|\mF|} = \lim_{q\to +\infty} \frac{q\cdot q^{k(mn-k)-1}}{q^{k(mn-k)}}
=1.\end{equation*} 
Therefore
\begin{equation} \label{concl2}
\lim_{q \to +\infty} \frac{q \cdot  \Lambda(q;mn,m(d-1),k)}{|\mF|} \cdot \left(1- \frac{\left(q^k-1\right)\left(q^{mn-k}-1\right)}{2\left(q^{mn}-q^{mn-k} \right)}\right)=1 \left(1-\frac{1}{2}\right)=\frac{1}{2}.
\end{equation}
The statement can now be obtained by combining (\ref{concl1}) and (\ref{concl2}).
\end{proof}


\begin{remark} \label{rem:commM}
 In the notation of the proof of Theorem~\ref{bigthm}, there are other choices of $M$ that would still yield a positive lower bound for $\lim_{q \to +\infty} |\mF'|/|\mF|$.
  Suppose that $M$ has the form $M=\alpha q^r$ for some $\alpha >0$ and $r\geq 0$. Substituting into (\ref{eq:ineq1}), we get 
  \begin{equation*} \label{concl3}
  \frac{|\mF'|}{|\mF|} \ge
  \frac{\alpha q^r \cdot \Lambda}{|\mF|} \left( 1- 
  \frac{\left(q^k-1\right)\left(q^{mn-k}-1\right)(\alpha q^r-1)}{\left(q^{mn}-q^{mn-k} \right)(q-1)}\cdot\frac{1}{2}\right). 
  \end{equation*}
  The right-hand side converges to a non-negative limit as $q \to + \infty$ if and only if $0\leq r\leq 1$. Now if $0\leq r<1$, then 
  $$\lim_{q \to +\infty} \frac{\alpha q^r \cdot \Lambda}{|\mF|} \left( 1- 
  \frac{\left(q^k-1\right)\left(q^{mn-k}-1\right)(\alpha q^r-1)}{\left(q^{mn}-q^{mn-k} \right)(q-1)}\cdot\frac{1}{2}\right) 
  = \lim_{q \to +\infty} \alpha q^{r-1}\left( 1- \frac{\alpha q^{r-1}}{2}\right) =0.$$
  If $r=1$, then
  $\lim_{q \to +\infty} |\mF'|/|\mF| = \alpha \left( 1- \alpha/2\right),$
  and the lower bound is maximal for $\alpha=1$.  
\end{remark}

We now compute the limit for the same formula as a function of $m$.  

\begin{corollary} \label{lim_m}
Fix an integer $d$ with $2 \le d \le n$, and let $k:=m(n-d+1)$. Define the families
$\mF:=\{\mC \le \mat \mid \dim(\mC)=k\}$ and $\mF':=\{\mC \in \mF \mid \mC \mbox{ is not MRD}\}$.
Then for every $\varepsilon \in \R_{>0}$ there exists $m_\varepsilon \in \N$ such that 
$$\DD{\mF'}{\mF} \ge \frac{1}{2} \left( \frac{q}{q-1}-\frac{1}{(q-1)^2} \right) -\varepsilon$$ for all integers $m \ge m_\varepsilon$.
In particular, 
$\lim_{m \to \infty} \DD{\mF'}{\mF} \ge 1/2$, provided the limit exists. 
\end{corollary}

\begin{proof}
To simplify notation, throughout the proof we let $N=mn$ and $t=N-k=m(d-1)$ . 
By Proposition~\ref{pro:inters} we have
$$\Lambda:=\Lambda(q;mn,m(d-1),k)=\Lambda(q;N,t,k)=
\sum_{h=1}^t \qbin{t}{h}{q} \sum_{s=h}^t \qbin{t-h}{s-h}{q} \qbin{N-s}{t}{q} (-1)^{s-h} q^{\binom{s-h}{2}}.$$
Moreover, from the proof of the proposition we see that
\begin{equation}
\sum_{s=h}^t \qbin{t-h}{s-h}{q} \qbin{N-s}{t}{q} (-1)^{s-h} q^{\binom{s-h}{2}} \ge 0 \label{qty} \quad \mbox{for all } 1 \le h \le t.
\end{equation}
Indeed, the quantity in (\ref{qty}) equals $|\{\mC \le X \mid \dim(\mC)=k, \ \mC  \cap \mD =\mH\}|$,
where $\mD,\mH \le \mat$ are any spaces with $\mH \le \mD$,
$\dim(\mD)=t$ and $\dim(\mH)=h$.
In particular,
\begin{eqnarray}
\Lambda \cdot \qbin{t}{1}{q}^{-1} &\ge&  \sum_{s=1}^t \qbin{t-1}{s-1}{q} \qbin{N-s}{t}{q} (-1)^{s-1} q^{\binom{s-1}{2}} \nonumber\\
&=& \qbin{N-1}{t}{q} - \qbin{t-1}{1}{q} \qbin{N-2}{t}{q} + \sum_{s=3}^t \qbin{t-1}{s-1}{q} \qbin{N-s}{t}{q}(-1)^{s-1} q^{\binom{s-1}{2}}. \ \ \ \ \ \label{extraterm}
\end{eqnarray}

We now show that the sum
$$\sigma(t)=\sum_{s=3}^t \qbin{t-1}{s-1}{q} \qbin{N-s}{t}{q}(-1)^{s-1} q^{\binom{s-1}{2}}=\sum_{s=3}^{\min\{t,N-t\}} \qbin{t-1}{s-1}{q} \qbin{N-s}{t}{q}(-1)^{s-1} q^{\binom{s-1}{2}},$$
is non-negative.
If $t \in \{2,3\}$ or $N-t \in \{2,3\}$, then the result is straightforward. Therefore we assume $t \ge 4$ and $N-t\ge 4$. Define $\overline{t}:=\max\{s \in \N \mid s \le \min\{t, N-t\}, \ s \mbox{ is even}\}$. Note that $\min\{t,N-t\}-\overline{t} \in \{0,1\}$.
Then we have
\begin{eqnarray}\sigma(t) &\ge&  \sum_{s=3}^{\overline{t}} \qbin{t-1}{s-1}{q} \qbin{N-s}{t}{q}(-1)^{s-1} q^{\binom{s-1}{2}} \nonumber \\ &=& \sum_{\substack{s \in \{3,...,\overline{t}-1\} \\ s \textnormal{ odd}}} \qbin{t-1}{s-1}{q} \qbin{N-s}{t}{q} q^{\binom{s-1}{2}} - \qbin{t-1}{s}{q} \qbin{N-s-1}{t}{q} q^{\binom{s}{2}}. \label{ssi}
\end{eqnarray}
By definition of $q$-binomial coefficient, for all $3 \le s \le \overline{t}-1$ we have
\begin{equation}\qbin{t-1}{s-1}{q} \qbin{N-s}{t}{q}q^{\binom{s-1}{2}}=
\qbin{t-1}{s}{q} \qbin{N-s-1}{t}{q} q^{\binom{s}{2}} \cdot \frac{(q^s-1)(q^{N-s}-1)q^{-s+1}}{(q^{t-s+1}-1)(q^{N-s-t}-1)},
\label{fract}\end{equation}
where the fraction on the right-hand side of (\ref{fract}) is well-defined by our choice of $\overline{t}$.
We have $q^r-1 \ge q^r/2$ for all $r \in \N$ with $r \ge 1$. Therefore
$$\frac{(q^s-1)(q^{N-s}-1)q^{-s+1}}{(q^{t-s+1}-1)(q^{N-s-t}-1)} \ge \frac{q^s}{4} \ge \frac{q^3}{4} > 1 \quad \mbox{ for all } 3 \le s \le \overline{t}-1.$$
In particular, from (\ref{fract}) we deduce that
\begin{equation}\qbin{t-1}{s-1}{q} \qbin{N-s}{t}{q}q^{\binom{s-1}{2}} \ge 
\qbin{t-1}{s}{q} \qbin{N-s-1}{t}{q} q^{\binom{s}{2}} \quad \mbox{ for all } 3 \le s \le \overline{t}-1.
\label{frac2}
\end{equation}
The fact that $\sigma(t)\ge 0$ now follows combining (\ref{ssi}) and (\ref{frac2}).

Using the above result and the bound in (\ref{extraterm}) we find
\begin{eqnarray}
\frac{q \cdot \Lambda}{|\mF|} &\ge& q \cdot {\qbin{N}{t}{q}}^{-1} \left( \qbin{t}{1}{q} \qbin{N-1}{t}{q}- \qbin{t}{1}{q} \qbin{t-1}{1}{q} \qbin{N-2}{t}{q}\right) \nonumber \\
&=& \frac{q(q^t-1)(q^{N-t}-1)}{(q-1)(q^N-1)}-
\frac{q(q^t-1)(q^{t-1}-1)(q^{N-t-1}-1)(q^{N-t}-1)}{(q-1)^2(q^N-1)(q^{N-1}-1)}. \label{lastexpr}
\end{eqnarray}
Denote by $\beta(m)$ the expression in (\ref{lastexpr}). Then applying Theorem~\ref{bigthm} one gets
$$|\mF'|/|\mF| \ge \beta(m) \cdot \left( 1- 
\frac{\left(q^k-1\right)\left(q^{mn-k}-1\right)}{2\left(q^{mn}-q^{mn-k} \right)}\right).$$
It is easy to check that, since $k\leq m(n-1)$, we have
$$\lim_{m \to +\infty} \beta(m) \cdot \left( 1- 
\frac{\left(q^k-1\right)\left(q^{mn-k}-1\right)}{2\left(q^{mn}-q^{mn-k} \right)}\right) =\frac{1}{2} \left(\frac{q}{q-1} - \frac{1}{(q-1)^2} \right),$$
from which the result follows.
\end{proof}

\section{Density of Maximal Codes}\label{sec:densmaxl}

Maximal codes are those that are not contained any larger code with the same minimum distance. Any extremal code with respect to the minimum distance
is therefore maximal, but the converse does not necessarily hold. However, as we show in this section, 
for $q$ sufficiently large the MDS codes are dense in the set of $\F_q$-linear maximal codes, and for $m$ sufficiently large
the vector MRD codes are dense in the set of $\F_{q^m}$-linear maximal codes.

\begin{definition} \label{def:maxcode}
An $\F_Q$-linear code $\mC \le X$ is called \textbf{maximal} if $\mC \neq \{0\}$, and there is no code $\mC' \le X$ such that
$\mC' \ge \mC$, $\dim(\mC')=\dim(\mC)+1$ and $d(\mC')=d(\mC)$.
\end{definition}

The following is an immediate consequence of Theorem~\ref{th_dist1}. We will shortly apply it for both Hamming metric and vector rank metric codes.

\begin{corollary}\label{cor:max1}
Let $\mC \le X$ be an $\F_Q$-linear code of dimension $k<N$ and minimum distance $d$. The number of $\F_Q$-linear codes in $\mC' \le X$ such that $\mC \le \mC'$, $\dim(\mC')=k+1$, and $d(\mC') = d$ is at least
	\begin{equation*}
	     \qbin{n-k}{1}{Q} \left(1-\frac{Q^{k}}{Q^{n}-Q^{k}} (\ball(d-1)-1)\right).
	\end{equation*}
\end{corollary}

\begin{proposition}
The following hold.
\begin{enumerate}
	\item If $q$ is sufficiently large, all maximal codes $\mC \le \F_q^n$ with respect to the Hamming metric are MDS.
	Equivalently, if $\mF$ is the family of $\F_{q}$-linear maximal codes in $\F_q^n$, and $\mF'$ is the family of MDS codes in 
	$\F_q^n$, then $\DD{\mF'}{\mF}=1$ for $q$ sufficiently large.
	\item If $m$ is sufficiently large, all maximal codes $\mC \le \F_{q^m}^n$ with respect to the rank metric are MRD.
	Equivalently, if $\mF$ is the family of $\F_{q^m}$-linear maximal codes in $\F_{q^m}^n$, and $\mF'$ is the family of MRD codes in 
	$\F_{q^m}^n$, then $\DD{\mF'}{\mF}=1$ for $m$ sufficiently large.
\end{enumerate}
 
\end{proposition}
\begin{proof}
Let $\mC \le \F_q^n$ be a non-zero maximal code. Denote by
 $k$ and $d$ the dimension and the minimum Hamming distance of $\mC$, respectively. Towards a contradiction, assume that
$d \le n-k$, so in particular, $\mC$ is not MDS. 
Combining Corollary~\ref{cor:max1} for the Hamming metric with the estimate in (\ref{estimate_ball_H}), we obtain 
$$\frac{q^k}{q^n-q^k} (\bH(d-1)-1) \in \mO\left( q^{k-n+d-1}\right) \mbox{ as } q \to +\infty.$$
Since $d \le n-k$, we have $k-n+d-1<0$. Therefore if $q$ is large, there exists at least one $ [n,k+1]_{q}$ code of minimum distance $d$ containing $\mC$, contradicting our choice of $\mC$ as maximal. 
We deduce that, for sufficiently large $q$, every maximal code is MDS. 

Similarly, if $\mC \le \F_{q^m}^n$ is a non-zero maximal code of $\F_{q^m}$-dimension $k$ and minimum rank distance $d\leq n-k$ over $\F_q$, then applying  Corollary~\ref{cor:max1} for the rank metric and the estimate in
(\ref{estimate_ball_rk}) we see that
$$\frac{q^{mk}}{q^{mn}-q^{mk}} (\brk(d-1)-1) \in \mO\left( q^{m(k-n+d-1)}\right) \mbox{ as } m \to +\infty.$$
As before, we deduce that there exists a $[n,k+1]_{q^m}$ code of minimum distance $d$ containing $\mC$, yielding a contradiction. Therefore, $d =n-k+1$ and $\mC$ is MRD. 
\end{proof}

\section{Density of Codes of Given Covering Radius}\label{denscovrad}

In this section we compute the density of codes of prescribed (often extremal) covering radius for Hamming and rank metric codes. These results do not use the machinery of partition-balanced families, but we include them for completeness.

\begin{definition}
Let $\mC \le X$ be an $\F_Q$-linear code, and let $r \in \N$. The 
\textbf{cloud} of radius $r$ about $\mC$ is the set
$$\Ball(\mC,r):= \bigcup_{x \in \mC} \Ball(x,r).$$
\end{definition}

The cloud about $\mC$ relates to the covering radius of $\mC$ as follows: For all $\rho \in \N$,  
$\rho(\mC) \le \rho$ if and only if $\Ball(\mC,\rho)=X$.
This implies the following simple preliminary result. 

\begin{lemma} \label{ovvio}
Let $\mC \le X$ be a linear code of dimension $k$, and let $\rho \ge 1$ be an integer. Assume
$q^k \cdot \ball(\rho-1)<q^N$. Then $\rho(\mC) \ge \rho$.
\end{lemma}

\subsection{Codes of Prescribed Dimension}

Recall from Section~\ref{sec:dist-reg} the well-known upper bounds for the covering radius of Hamming and vector rank metric codes in the form of the redundancy bounds:
$$\rhoH(\mC) \leq n-k \quad \text{ and } \quad \rhork(\mC) \leq n-k,$$
for codes $\mC$ of dimension $k$ over $\fq$ and $\F_{q^m}$, respectively. It is not hard to see that for $q$ and $m$ sufficiently large,
all codes meet these bounds, respectively, as we show shortly.

\begin{proposition} \label{prop:HammingCR}
Let $k$ be an integer with $0 \le k \le n$, and let $\mC \le \F_q^n$ be an $[n,k]_q$ code. We have
$$\rhoH(\mC) =n-k, \quad \mbox{provided that $q$ is  sufficiently large}.$$

Equivalently, if $\mF$ is the family of $[n,k]_q$ codes and 
$\mF':=\{\mC \in \mF \mid \rhoH(\mC) =n-k\}$, 
then $\DD{\mF'}{\mF}=1$ for $q$ sufficiently large.
\end{proposition}

\begin{proof}	
The result is clear if $k=n$, thus we assume $k<n$. By the redundancy bound of Proposition~\ref{Singleton_redun} we have $\rhoH(\mC) \le n-k$. Using the estimate in (\ref{estimate_ball_H}) one easily sees that
$q^k \cdot \bH(n-k-1) \in \mO(q^{n-1})$ as $q \to +\infty$ and so for $q$ sufficiently large, by Lemma~\ref{ovvio}, it holds that $\rhoH(\mC) \geq n-k$.
The result follows.
\end{proof}

We now show the analogous result for vector rank metric codes.

\begin{proposition} \label{prop:rkCR}
Let $k$ be an integer with $0 \le k \le n$, and let $\mC \le \F_{q^m}^n$ be an $[n,k]_{q^m}$ vector rank metric code.
We have
$$\rhork(\mC) =n-k, \quad \mbox{provided that $m$ is  sufficiently large}.$$

Equivalently, if $\mF$ is the family of $[n,k]_{q^m}$ codes and 
$\mF':=\{\mC \in \mF \mid \rhork(\mC) =n-k\}$, 
then $\DD{\mF'}{\mF}=1$ for $m$ sufficiently large.
\end{proposition}

\begin{proof}
The result is immediate if $k=n$. Now assume $k<n$. By the rank metric redundancy bound of Proposition~\ref{lem:rkred} we have $\rhork(\mC) \leq n-k$. 
From the estimate in (\ref{estimate_ball_rk}) it follows that $\brk(\rho) \in \Theta(q^{\rho m})$ as $m \to \infty$.
Thus $q^{mk} \cdot \brk(n-k-1) \in \mO(q^{m(n-1)})$ as $m \to + \infty$. Again by Lemma~\ref{ovvio}, the result follows.
\end{proof}

In the remainder of the section we focus on $\F_q$-linear matrix rank metric codes, and establish the following result on their covering radii.
\begin{theorem} \label{th_rk_in}
Let $k$ be an integer with $0 \le k \le nm$.
Denote by $\mF$ the family of $\F_q$-linear rank metric codes
of dimension $k$.
Define 
$\rho_k:=n-\lfloor k/m \rfloor$, and let
$\mF':=\{\mC \in \mF \mid \rhork(\mC) =\rho_k\}$.
Then
$$\lim_{q \to \infty} \DD{\mF'}{\mF} =1, \quad \mbox{ provided that } k < (m-n+\lfloor k/m \rfloor +1)(\lfloor k/m \rfloor +1).$$
\end{theorem}

Note that the requirement on $k$ in the statement above is automatically satisfied, for example, when $m \ge n + k$, $k < m$, or when 
$m > (t+1)(n-t-1)$, where $k=(t+\varepsilon)m$ and $0 \leq \varepsilon <1$.

Before presenting the proof of Theorem~\ref{th_rk_in}, we point out the following 
interesting difference between $\F_{q^m}$-linear and $\F_q$-linear rank metric codes.

\begin{remark}
 As the following example shows, not all $k$-dimensional codes
$\mC \le \mat$ have covering radius upper bounded by $n-\lfloor k/m \rfloor$. Therefore 
one cannot apply directly the proof strategy in Proposition~\ref{prop:rkCR} to establish 
Theorem~\ref{th_rk_in}. The idea behind our argument is to first show that \textit{most} $k$-dimensional codes
$\mC \le \mat$ have covering radius upper bounded by $n-\lfloor k/m \rfloor$.
\end{remark}

\begin{example}
Assume $m \ge 2n$ and $n(m-n) >m$. Fix an integer $k$ with $m< k \le n(m-n)$, and take a subset $L \subseteq \{1,...,n\} \times \{1,...,m-n\}$
of cardinality $k$.
Define $$\mC:=\{x \in \F_q^{n \times m} \mid x_{ij}=0 \mbox{ whenever } (i,j) \notin L\} \le \F_q^{n \times m}, \qquad y:= \begin{bmatrix} 0_{n \times (m-n)} & I_n\end{bmatrix} \in \F_q^{n \times m}.$$
The code $\mC$ is linear of dimension $k$, and all its codewords are matrices supported on $L$.
Therefore $\rk(x-y)=n$ for all $x \in \mC$. This shows that $\rho(\mC)=n$. Thus $n-\lfloor k/m \rfloor <n=\rho(\mC)$.
\end{example}

We now derive a matrix-analogue of the \textit{redundancy bound} for linear codes with the Hamming metric.
See~\cite{byrneravagnani} for a bound of the same type that takes into account also the minimum rank distance of the code.

\begin{definition}
Let $S \subseteq \{1,...,n\} \times \{1,...,m\}$ be a set. The
\textbf{characteristic matrix} of $S$ is the binary $n \times m$ matrix $\chi(S)$
defined by $\chi(S)_{ij}:=1$ if $(i,j) \in S$, and 
$\chi(S)_{ij}:=0$ otherwise.
We denote by $\lambda(S)$ the minimum numbers of lines (rows or columns) needed to cover all the 1's in the matrix $\chi(S)$.

The \textbf{initial entry} of a non-zero matrix $x \in \mat$ is $\inn(x):=\min\{(i,j) \mid x_{ij} \neq 0\}$, where the minimum is with respect to the lexicographic order. The 
\textbf{initial set} of a non-zero $\F_q$-linear rank metric code $\mC \le \mat$ is $$\inn(\mC):=\{\inn(x) \mid x \in \mC, \ x \neq 0\} \subseteq \{1,...,n\} \times 
\{1,...,m\}.$$
\end{definition}

\begin{proposition} \label{ISB}
Let $\mC \le \mat$ be a non-zero $\F_q$-linear rank metric code. Denote by 
$S$ the complement of $\inn(\mC)$ in $\{1,...,n\} \times \{1,...,m\}$. Then
$\rhork(\mC) \le \lambda(S)$.
\end{proposition}

\begin{proof}
Choose an arbitrary matrix $y \in \mat$. Let by $x$ the unique matrix in 
$\mC$ such that $x_{ij}=y_{ij}$ for all $(i,j) \in \inn(\mC)$.
Then $x-y$ is supported on $S$. In particular, 
$\rk(x-y) \le \lambda(S)$. Since $y$ was arbitrary, this shows that
$\rhork(\mC) \le \lambda(S)$.
\end{proof}

We are now ready to show the main result of this subsection. 

\begin{proof}[Proof of Theorem~\ref{th_rk_in}]
The result is trivial if $k=0$. From now on we assume $0<k<mn$.
Let $J_k \subseteq \{1,...,n\} \times \{1,...,m\}$ be the set of the first 
$k$ elements of $\{1,...,n\} \times \{1,...,m\}$, with respect to the lexicographic order.   
Define the familiy of codes $
\mF'':=\{\mC \in \mF \mid \inn(\mC) =J_k\}$.

Using Proposition~\ref{ISB} one can easily check that
$\rhork(\mC) \le \rho_k$ for all $\mC \in \mF''$.
Now assume that $\mC \in \mF$ is chosen uniformly at random. Then 
\begin{eqnarray}
\Pro[\rhork(\mC)=\rho_k] &=& \Pro[\rhork(\mC) > \rho_k-1] -  \Pro[\rhork(\mC) > \rho_k] \nonumber \\
&=&  \Pro[\rhork(\mC) > \rho_k-1] -  1+\Pro[\rhork(\mC) \le \rho_k] \nonumber \\
&\ge& \Pro[\rhork(\mC) > \rho_k-1] -  1 + \Pro[\mC \in \mF'']. \label{aaaa1}
\end{eqnarray} 
It is easy to see that 
\begin{equation*} 
\Pro[\mC \in \mF''] =q^{k(mn-k)} \qbin{mn}{k}{q}^{-1}.
\end{equation*}
Therefore
\begin{equation} \Pro[\mC \in \mF'']-1  = \frac{\displaystyle q^{k(mn-k)}\prod_{j=0}^{k-1}(q^k-q^j) - \prod_{j=0}^{k-1}(q^{mn}-q^j)}{\displaystyle \prod_{j=0}^{k-1}(q^{mn}-q^j)} \to 0 \quad \mbox{as } q \to +\infty,\label{aaaa3}
\end{equation}
where the latter estimate follows from the fact that 
the leading term in $q$ of the numerator has degree strictly less than ${mnk}$, which is the degree in $q$ of the denominator. 
Finally, using the estimate in (\ref{estimate_ball_rk}) we obtain 
$$q^{k-mn} \cdot \brk(\rho_k-1) \in \Theta \left(q^{k-(m-n+\lfloor k/m \rfloor +1))(\lfloor k/m \rfloor +1)} \right) \mbox{ as } 
q \to +\infty.$$
This vanishes as  $q \to \infty$ when $k < (m-n+\lfloor k/m \rfloor +1)(\lfloor k/m \rfloor +1).$
Thus under this assumption we have
 $\Pro[\rhork(\mC) \ge \rho_k] =  \Pro[\rhork(\mC) > \rho_k-1] =1$ for $q$ sufficiently large.
The theorem can now be obtained combining this observation with (\ref{aaaa1}) and (\ref{aaaa3}).
\end{proof}

\subsection{Maximal Codes}

We explicitly compute the covering radius of maximal codes (see Definition~\ref{def:maxcode}) endowed with the Hamming and rank metric, for sufficiently large values of some of the code parameters. 

Note that if $\mC \le X$ is a code of minimum distance $d$ and $\rho(\mC) \ge d$, then there exists $y \in X \setminus \mC$ such that $d(x,y) \ge d$ for all $x \in \mC$. Therefore $\mC \cup \{y\}$ is a (non-linear) code of minimum distance $d$ that contains $\mC$. This observation is known as the \textit{supercode lemma} in the literature~\cite{coveringcodes}. In the sequel, we will need the following ``linear'' version of this result for distance-regular spaces, which relies on Properties~\ref{pr3} and~\ref{pr3a} of Definition~\ref{defregular}.

\begin{lemma}[Linear Supercode Lemma] \label{LSC}
Let $\mC \le X$ be a non-zero code of minimum distance $d$. If $\mC$ is maximal, then $\rho(\mC) \le d-1$.
\end{lemma}
\begin{proof}
Assume by contradiction that there exists $y \in X \setminus \mC$ with $d(x,y) \ge d$ for all $x \in \mC$. Define 
the linear code $\mC':=\mC + \langle y \rangle \gneq \mC$. Let $y_1:=x_1+\alpha_1 y$ and $y_2:=x_1+\alpha_2 y$ be arbitrary distinct elements of $\mC'$, with 
 with $x_1,x_2 \in \mC$ and $\alpha_1,\alpha_2 \in \F_q$.  We will show that
$d(y_1,y_2) \ge d$.
By Property~\ref{pr3a} of Definition~\ref{defregular} we have
$$d(y_1,y_2)= d(x_1-x_2,(\alpha_2-\alpha_1)y).$$
If $\alpha_1=\alpha_2$, then we must have $x_1 \neq x_2$, as $y_1$ and $y_2$ are distinct. Therefore, again by Property~\ref{pr3a} of Definition~\ref{defregular},  $d(y_1,y_2)=d(x_1,x_2) \ge d$.
On the other hand, if $\alpha_1 \neq \alpha_2$ then by Property~\ref{pr3} of Definition~\ref{defregular} we conclude
$d(y_1,y_2)=d((\alpha_2-\alpha_1)^{-1}(x_1-x_2),y) \ge d$, where the last inequality follows from the linearity of $\mC$. This contradicts the maximality of $\mC$.
\end{proof}

\begin{proposition}
Let $d$ be an integer with $2 \le d \le n$, and let $\mC \le \F_q^n$ be
a linear maximal code of minimum Hamming distance $d$. Then
$$\rhoH(\mC)=d-1, \quad \mbox{provided that $q$ is sufficiently large}.$$

Equivalently, if $\mF$ denotes the family of linear maximal codes $\mC \le \F_q^n$ with $\dH(\mC)=d$, and 
$\mF':=\{\mC \in \mF \mid \rhoH(\mC)=d-1\}$, then
$\DD{\mF'}{\mF}=1$ for $q$ sufficiently large.
\end{proposition}

\begin{proof}
Since $\mC$ is maximal, we have $\rhoH(\mC) \le d-1$ by Lemma~\ref{LSC}.
Let $k$ denote the dimension of $\mC$. By Proposition~\ref{Singleton_redun} we have
$q^k \bH(d-2) \le q^{n-d+1}\bH(d-2) \in \mO(q^{n-1})$
as $q \to +\infty$, where the asymptotic estimate follows from (\ref{estimate_ball_H}). Therefore $\rhoH(\mC) \ge d-1$ when
 $q$ is sufficiently large.
\end{proof}

We now focus on rank metric codes, and explicitly compute the covering radius of any maximal code $\mC \le \mat$, for $m$ sufficiently large and for any $q$.

\begin{theorem} \label{CRmax}
Let $\varepsilon$ be a positive integer, and let $\mC \le \mat$ be a linear non-zero matrix rank metric code with minimum rank distance $d \ge \varepsilon$. Assume 
$(\varepsilon-1)m \ge \log_q(4)+n^2/4$. Then
$$\rhork(\mC) \ge d-\varepsilon+1.$$ 

Equivalently, if $\mF$ denotes the family of non-zero linear rank metric codes $\mC \le \mat$ with $\drk(\mC) \ge \varepsilon$, and 
$\mF':=\{\mC \in \mF \mid \rhoH(\mC) \ge d-\varepsilon+1\}$, then we have
$\DD{\mF'}{\mF}=1$, provided that $(\varepsilon-1)m \ge \log_q(4)+n^2/4$.
\end{theorem}

\begin{proof}
The size of a ball of radius 
$d-\varepsilon$ in the rank metric space $\mat$ can be upper bounded (see e.g. \cite[Lemma 5]{gad_paper}) as
$$\brk(d-\varepsilon) < q^{(d-\varepsilon)(n+m-d+\varepsilon)} \cdot \prod_{j=1}^{+\infty} \frac{1}{1-q^{-j}}.$$
It is well known (see~\cite{berlekamp}) that 
$\prod_{j=1}^{+\infty} \frac{1}{1-q^{-j}} < 4$ for any prime power $q$. Therefore by Proposition~\ref{rankSing} 
we have
\begin{equation} \label{qqq1}
|\mC| \cdot \brk(d-\varepsilon) < 4 \cdot q^{m(n-d+1)} \cdot q^{(d-\varepsilon)(n+m-d+\varepsilon)}.
\end{equation}
It is easy to see that
\begin{equation} \label{qqq2}
\log_q(4) + m(n-d+1)+(d-\varepsilon)(n+m-d+\varepsilon) \le nm \quad \mbox{if } (\varepsilon-1)m \ge \log_q(4)+n^2/4.
\end{equation}
Indeed, the inequality in (\ref{qqq2}) holds if and only if
$(\varepsilon-1)m \ge \log_q(4) -d^2+d(n+2\varepsilon) -\varepsilon n-\varepsilon^2$. Moreover,
the real-valued function $d \mapsto \log_q(4) -d^2+d(n+2\varepsilon) -\varepsilon n-\varepsilon^2$  attains 
its maximum for $d=n/2+\varepsilon$. 
Combining (\ref{qqq1}) and (\ref{qqq2}) we therefore obtain
$$|\mC| \cdot \brk(d-\varepsilon) < q^{nm} \quad \mbox{whenever } \ (\varepsilon-1)m \ge \log_q(4)+n^2/4.$$
Then by Lemma~\ref{ovvio} we have $\rhork(\mC) \geq  d-\varepsilon+1$. 
The second part of the statement follows from Lemma~\ref{LSC}.
\end{proof}

\begin{corollary}
Let $d$ be an integer with $2 \le d \le n$.
Denote by $\mF$ the family of maximal codes $\mC \le \mat$ with $\drk(\mC) \ge 2$, and 
let $\mF':=\{\mC \in \mF \mid \rhoH(\mC) \ge d-1\}$. Then
$\DD{\mF'}{\mF}=1$, provided that $(\varepsilon-1)m \ge \log_q(4)+n^2/4$.
\end{corollary}

\section{Average Parameters of Codes}\label{sec:avpar}

In this final section we show a second class of applications 
of partition-balanced families, namely, the computation of certain average parameters of codes.

We first compute the average weight distribution of a code of given dimension, and compare it the weight distribution of MDS and MRD codes. We then study the average distance of a code from a given element $x \in X$. This parameter generalizes the well-known 
\textit{total weight} of a code.

\begin{definition}
A \textbf{code parameter} is a function $p: \{\mbox{codes in } X\} \to \R \cup \{+\infty\}$.
If $\mF$ is a non-empty family of codes in $X$, then we 
denote the {average value} of $p$ on $\mF$ by
$$\overline{p}(\mF):= |\mF|^{-1} \sum_{\mC \in \mF} p(\mC).$$
\end{definition}

An example of a code parameter is, for $i \in \N$, the $i$-th element of the weight distribution, i.e.,
$\mC \mapsto W_i(\mC)$. Another example is 
the {average distance from} $x$, i.e., $\mC \mapsto \delta_x(\mC):=|\mC|^{-1} \sum_{y \in \mC} d(x,y)$.
In the remainder of the section we focus on these two parameters.

\begin{theorem}
Let $0 \le k \le N$ be an integer, and let $\mF$ denote the family of $k$-dimensional linear codes in $X$. The average weight distribution of $\mF$ is 
$$\overline{W_0}(\mF)=1 \qquad \mbox{and} \qquad \overline{W_i}(\mF)=\frac{q^k-1}{q^N-1} \ W_i(X) \quad \mbox{ for } 1 \le i \le |\omega|.$$
\end{theorem}

\begin{proof}
The result is clear if $i=0$. When $i>0$, we can apply Lemma ~\ref{KEY} to the family of $k$-dimensional codes in $X$ and to the partition of Proposition~\ref{examp1}, with $\mD=\{0\}$. We take as $f$ the characteristic function of $\Ball(0,i)\setminus \Ball(0,i-1)$. 
\end{proof}

\begin{corollary}[Corollary 3.1.21 of~\cite{pelbook}] \label{avgdistr1}
Let $k$ be an integer with $0 \le k \le n$. Let $\mF$ be the family of $[n,k]_q$ codes endowed with the Hamming metric. 
The average Hamming weight distribution of $\mF$ is 
$$\overline{\WH_0}(\mF)=1 \qquad \mbox{and} \qquad \overline{\WH_i}(\mF)=\frac{q^k-1}{q^n-1} \ \binom{n}{i}(q-1)^i \quad \mbox{ for } 1 \le i \le n.$$
\end{corollary}

\begin{corollary} \label{avgdistr2}
Let $k$ be an integer with $0 \le k \le n$. Let $\mF$ be the family of $[n,k]_{q^m}$ vector rank metric codes. 
The average rank weight distribution of $\mF$ is 
$$\overline{\Wrk_0}(\mF)=1 \qquad \mbox{and} \qquad \overline{\Wrk_i}(\mF)=\frac{q^{mk}-1}{q^{mn}-1} \ \qbin{n}{i}{q}\prod_{j=0}^{i-1}(q^m-q^j) \quad \mbox{ for } 1 \le i \le n.$$
\end{corollary}

\begin{corollary} \label{avgdistr3}
Let $k$ be an integer with $0 \le k \le mn$. Let $\mF$ be the family of $k$-dimensional matrix rank metric codes in $\F_q^{n \times m}$. 
The average rank weight distribution of $\mF$ is 
$$\overline{\Wrk_0}(\mF)=1 \qquad \mbox{and} \qquad \overline{\Wrk_i}(\mF)=\frac{q^{k}-1}{q^{mn}-1} \ \qbin{n}{i}{q}\prod_{j=0}^{i-1}(q^m-q^j) \quad \mbox{ for } 1 \le i \le n.$$
\end{corollary}

\begin{remark}
It is interesting to observe that, although the family of MDS $[n,k]_q$ codes is dense in the family of $[n,k]_q$ codes as $q\to +\infty$, the $(d-1)$-th component of the average Hamming weight distribution of an $[n,k]_q$ code does not converge to $0$ 
for $q \to +\infty$. Following the notation of Corollary~\ref{avgdistr1}, assume $1 \le k <n$ and let $d=n-k+1$.
Then we have
\begin{equation*}
\lim_{q \to \infty} \overline{\WH_{i}}(\mF) = \left\{ \begin{array}{cl} 0 & \mbox{ if } 1 \le i \le d-2, \\ \binom{n}{d-1}
& \mbox{ if } i=d-1, \\ +\infty & \mbox{ if } d \le i \le n. \end{array}\right.
\end{equation*}
In particular,
\begin{equation} \label{neq0}
\lim_{q \to \infty} \overline{\WH_{d-1}}(\mF) = \binom{n}{d-1} \neq 0.
\end{equation}
Similarly, the $(d-1)$-th component of the average rank weight distribution of an $[n,k]_{q^m}$ code does not converge to 
$0$ for $m \to +\infty$. In the notation of Corollary~\ref{avgdistr2}, assume $1 \le k <n$, and let $d:=n-k+1$. Then
$$\lim_{m \to \infty} \overline{\Wrk_{d-1}}(\mF) = \qbin{n}{d-1}{q} \neq 0,$$
which is the rank-analogue of the limit in (\ref{neq0}).
Note that if $\mF$ is instead the family of $\F_q$-linear matrix codes of dimension $k$, as in Corollary~\ref{avgdistr3}, then
$$\lim_{q \to \infty} \overline{\Wrk_{d-1}}(\mF) = +\infty.$$
\end{remark}

\bigskip

The following result for codes in Hamming space $\F_q^n$ is well-known.
Let $\mC$ be a linear code in $\F_{q}^n$ and let $x \in \F_{q}^n$. Recall that the \textbf{support} of $\mC$ is  $\mathrm{supp}(\mC) = \{1 \le i \le n \mid y_i \neq 0 \text{ some } y \in \mC\}$. 
Then 
$$\delta_x(\mC)= \frac{1}{|\mC|} \sum_{y \in \mC} \dH(x,y) = \frac{q-1}{q} \ |\mathrm{supp}(\mC)| + \wH(\hat{x}),$$
where $\hat{x}$ is the vector of length $n-|\mathrm{supp}(\mC)|$ obtained by puncturing $x$ on the support of $\mC$.
In particular, if $\mC$ has support $n$, then the average distance of a vector $x$ to the codewords of $\mC$ is given by the constant $\delta_x(\mC)=n (q-1)/q$. 
This simple result has found various applications in coding theory~\cite{bgks10, gabklove}. However no such result holds for the rank distance. The average of the ranks of the elements of the coset $x+\mC$ of a matrix code is neither an invariant with respect to $\mC$, nor with respect to the rank of $x$. 

However, the average distance of an element $x \in X$ to a member of a {\em family} of (partition balanced) linear codes is often an invariant of the parameters of that family and the weight of $x$. We demonstrate this with the following results.

\begin{theorem} \label{OKfirsttheo}
Let $k$ be an integer with $0 \le k \le N$, and let $\mF$ be the family of linear codes
$\mC \le X$ with $\dim(\mC)=k$. For all $x \in X$ we have
$$\overline{\delta_x}(\mF) = \frac{\omega(x)}{q^k} + \frac{q^k-1}{q^k(q^N-1)} \left( -\omega(x)+\sum_{i=0}^{|\omega|} i \cdot \p{0}{i}{i} \right).$$ 
\end{theorem}

\begin{proof}
Define the function $f:X \to \N$ by $f(y):= d(x,y)$ for all
$y \in X$.
We apply Lemma~\ref{KEY} to the family $\mF$ and the function $f$, taking as $\mP$ the partition of Proposition~\ref{examp1}. We obtain
\begin{eqnarray*}
\sum_{\mC \in \mF} \sum_{y \in \mC} d(y,x) &=& 
|\mF| \cdot \omega(x) + \frac{|\mF| \cdot (q^k-1)}{q^N-1} \sum_{y \in X \setminus \{0\}} d(x,y) \\
&=& |\mF| \cdot \omega(x) + \frac{|\mF| \cdot (q^k-1)}{q^N-1} \left( -\omega(x)+\sum_{y \in X} d(x,y) \right) \\
&=& |\mF| \cdot \omega(x) + \frac{|\mF| \cdot (q^k-1)}{q^N-1} \left( -\omega(x)+\sum_{i=0}^{|\omega|} i \cdot \p{0}{i}{i} \right).
\end{eqnarray*}
The result follows.
\end{proof}

We now establish an analogous result for a
$\mP(\omega)$-balanced family $\mF$, under the assumption that
all codes in $\mF$ have the same weight distribution.

\begin{theorem} \label{OKavgPomega}
Let $k$ and $d$ be integers with $1 \le k \le N$ and $1 \le d \le |\omega|$. Let $\mF$ be a $\mP(\omega)$-balanced family of linear codes in $X$ of dimension $k$ and minimum distance $d$. Assume that all codes in $\mF$ have the same weight distribution, say
$(W_i(\mF) \mid i \in \N)$. For all $x \in X$ we have
$$\overline{\delta_x}(\mF) =q^{-k} \left(\omega(x) + \sum_{i=d}^{|\omega|} \frac{W_i(\mF)}{W_i(X)} \sum_{j=0}^{|\omega|} j \cdot \p{\omega(x)}{i}{j}\right).$$ 
\end{theorem}

\begin{proof}
Define the function $f:X \to \N$ by $f(y):= d(x,y)$ for all
$y \in X$. As in the proof of Theorem~\ref{OKfirsttheo}, we apply Lemma~\ref{KEY} to the partition $\mP(\omega)$ of Proposition~\ref{examp2} and obtain
$$\overline{\delta_X}(\mF)  \cdot q^k = |\mF|^{-1} \sum_{\mC \in \mF}  \sum_{y \in \mC} d(x,y) = \omega(x) + \sum_{i=d}^{|\omega|} \frac{W_i(\mF)}{W_i(X)} \sum_{\substack{y \in X \\ \omega(y)=i}} d(x,y).$$
By definition of intersection numbers, we have
$$\sum_{\substack{y \in X \\ \omega(y)=i}} d(x,y) = \sum_{j=0}^{|\omega|} j \cdot \p{\omega(x)}{i}{j},$$
which gives the desired formula.
\end{proof}

Theorems~\ref{OKfirsttheo} and~\ref{OKavgPomega} can specialized to codes with the Hamming and the rank metric. 
We apply these results in the case of the Hamming metric, using the formula derived in Section~\ref{sec:dist-reg} for the intersection numbers.

\begin{example}
 Let $k$ be an integer with $1 \le k \le n$, and let $\mF$ be the family of MDS $[n,k]_q$ codes in the Hamming space $\F_{q}^n$. Then $\mF$ is balanced with respect to the partition on $\F_q^n$ induced by the Hamming weight, and each code in $\mF$ has the same weight distribution \cite{macws}. In the notation of Theorem~\ref{OKavgPomega}, we have
   	\begin{eqnarray*}
   	W_i(\F_{q}^n)&=&\pH{0}{i}{i} =\binom{n}{i}(q-1)^i \\
   	W_i(\mF)& =& \binom{n}{i}(q-1)\sum_{j=0}^{i-d}(-1)^j\binom{i-1}{j}q^{i-d-j} \\
   	\pH{\omega(x)}{i}{j}& = & \displaystyle \sum_{r\geq 0}\binom{\omega(x)}{r}\binom{n-\omega(x)}{j-\omega(x)+r}\binom{\omega(x)-r}{j-i+r}(q-1)^{j-\omega(x)+r}(q-2)^{i+\omega(x)-j-2r},
   \end{eqnarray*}
   where $\omega=\wH$ is the Hamming weight.
Applying Theorem~\ref{OKavgPomega}, we get
\begin{multline*}
    \overline{\delta_x}(\mF) \cdot q^k = \omega(x) + \sum_{i=d}^{n} \frac{\sum_{j=0}^{i-d}(-1)^j\binom{i-1}{j}q^{i-d-j}}{(q-1)^{i-1}} \sum_{j=0}^{n} j  \sum_{r}\binom{\omega(x)}{r} \cdots \\ \cdots \binom{n-\omega(x)}{j-\omega(x)+r}\binom{\omega(x)-r}{j-i+r}(q^{i-r}+\cdots).
\end{multline*}
It can be checked that 
\begin{eqnarray*}
\lim_{q \to +\infty}\overline{\delta_x}(\mF)  &=& \lim_{q \to +\infty} q^{-k}\left(\omega(x) + q^{-d+1} q^n \sum_{j=\omega(x)}^{n} j \binom{n-\omega(x)}{j - \omega(x)}\binom{\omega(x)}{j-n}\right)
\\ &=&  \lim_{q \to +\infty} nq^{n-d+1-k} = n.
\end{eqnarray*}
\end{example}

\end{document}